\documentclass[conference]{IEEEtran}
\usepackage[ruled,vlined]{algorithm2e}
\usepackage{cite}
\usepackage{graphicx}
\usepackage{psfrag}
\usepackage{subfigure}
\usepackage{url}
\usepackage{amsmath}
\usepackage{yhmath}
\usepackage{array}
\usepackage{amssymb}
\usepackage{amsfonts}
\usepackage{graphicx}
\usepackage{epstopdf}
\usepackage{multirow}

\newtheorem{lemma}{Lemma}

\newtheorem{corollary}{Corollary}

\newtheorem{definition}{Definition}
\newtheorem{theorem}{Theorem}
\newtheorem{example}{Example}

\title{Wireless Network-Coded Multi-Way Relaying Using Latin Hyper-Cubes}
\begin{document}

\author{
\authorblockN{Srishti Shukla and B. Sundar Rajan}\\
\authorblockA{Email: {$\lbrace$srishti, bsrajan$\rbrace$} @ece.iisc.ernet.in\\
IISc Mathematics Initiative (IMI), Dept. of Mathematics and Dept. of Electrical Comm. Engg., IISc, Bangalore\\
}
}

\maketitle
\begin{abstract}
Physical layer network-coding for the $n$-way wireless relaying scenario is dealt with, where each of the $n$ user nodes $X_1,$ $X_2,...,X_n$ wishes to communicate its messages to all the other $(n-1)$ nodes with the help of the relay node R. The given scheme, based on the denoise-and-forward scheme proposed for two-way relaying by Popovski et al. in \cite{PoY1}, employs two phases: Multiple Access (MA) phase and Broadcast (BC) phase with each phase utilizing one channel use and hence totally two channel uses. Physical layer network-coding using the denoise-and-forward scheme was done for the two-way relaying scenario in\cite{KPT}, for three-way relaying scenario in \cite{SVR}, and for four-way relaying scenario in \cite{ShR}. This paper employs denoise-and-forward scheme for physical layer network coding of the $n$-way relaying scenario illustrating with the help of the case $n = 5$ not dealt with so far. It is observed that adaptively changing the network coding map used at the relay according to the channel conditions reduces the impact of multiple access interference which occurs at the relay during the MA phase. These network coding maps are chosen so that they satisfy a requirement called \textit{exclusive law}. We show that when the $n$ users transmit points from the same $M$-PSK $(M=2^{\lambda})$ constellation, every such network coding map that satisfies the exclusive law can be represented by a $n$-fold Latin Hyper-Cube of side $M$. The singular fade subspaces resulting from the scheme are described and enumerated for general values of $n$ and $M$ and are classified based on their removability in the given scenario. A network code map to be used by the relay for the BC phase aiming at reducing the effect of interference at the MA stage is obtained. 
\end{abstract} 

\vspace{-0.15 cm}
\section{Background And Preliminaries}
The two-stage protocol for physical layer network coding for the two-way relay channel first introduced in \cite{ZLL}, exploits the multiple access interference occurring at the relay so that the communication between the end nodes can be done using a two stage protocol. The works in \cite{KMT}, \cite{PoY} deal with the information theoretic studies for bidirectional relaying. In \cite{KPT}, modulation schemes to be used at the nodes for uncoded transmission for the two-way relaying were studied.

The work done for the relay channels with three or more user nodes is given in \cite{SVR,ShR,LiA,PiR,PaO,JKPL}. In \cite{LiA}, authors have proposed a two stage operation for three-way relaying called joint network and superposition coding, in which the three users transmit to the relay node one-by-one in the first phase, and the relay node makes two superimposed XOR-ed packets and transmits back to the users in the BC phase. The packet from the node with the worst channel gain is XOR-ed with the other two packets. The protocol employs four channel uses, three for the MA phase and one for the BC phase. It is claimed by the authors that this scheme can be extended to more than three users as well. In the work by Pischella and Ruyet in \cite{PiR} a lattice-based coding scheme combined with power control, composed of alternate MA and BC phases, consisting of four channel uses for three-way relaying is proposed. The relay receives an integer linear combination of the symbols transmitted by the user nodes. It is stated that the scheme can be extended to more number of users. These two works essentially deal with the information theoretic aspects of multi-way relaying. An `opportunistic scheduling technique' for physical network coding is proposed by authors Jeon et al. in \cite{JKPL}, where using a channel norm criterion and a minimum distance criterion, users in the MA as well as the BC phase are selected on the basis of instantaneous SNR. This approach utilizes six channel uses in case of three-way relaying and it is mentioned that the approach can be extended to more number of users. In \cite{PaO}, a `Latin square-like condition' for the three-way relay channel network code is proposed and cell swapping techniques on these Latin Cubes are suggested in order to improve upon these network codes. The protocol employs five channel uses, and the channel gains associated with the channels are not considered in the construction of this network coding map. 

\begin{figure}[tp]
\center
\includegraphics[height=35mm]{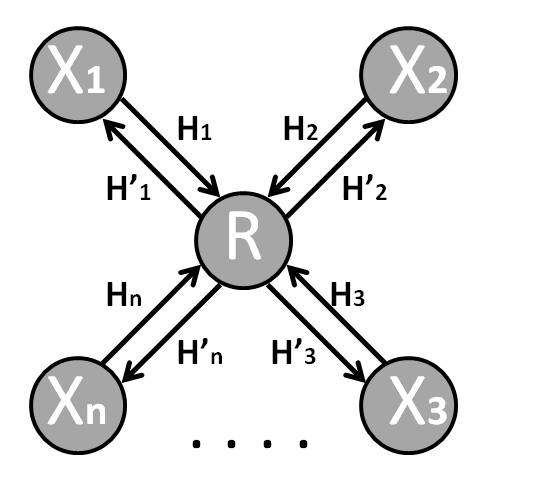}
\vspace{-.2 cm}
\caption{An $n$-way relay channel}
\vspace{-0.8 cm}
\end{figure}

We consider the $n$-way wireless relaying scenario shown in Fig. 1, where $n$-way data transfer takes place among the nodes $X_1$, $X_2$,..., $X_n$ with the help of the relay R assuming that the $n$ nodes operate in half-duplex mode. The relaying protocol consists of two phases, \textit{multiple access} (MA) phase, consisting of one channel use during which $X_1$, $X_2$,..., $X_n$ transmit to R; and \textit{broadcast} (BC) phase, in which R transmits to $X_1$, $X_2$,..., $X_n$ in a single channel use. Network Coding is employed at R in such a way that $X_i$ can decode $X_j$'s message for $i,j=1,2,...,n$ and $j \neq i$, given that $X_i$ knows its own message. Latin Cubes have been explored as a tool to find the network coding map used by the relay, depending on the channel gain in \cite{SVR}. The throughput performance of the two stage protocol for three-way relaying given in \cite{SVR} is better than the throughput performance of the `opportunistic scheduling technique' given in \cite{JKPL} at high SNR, as can be observed from the plots given in \cite{SVR}. The work in \cite{ShR} further extends the approach used in \cite{SVR} to four-way relaying and employs two channel uses for the entire information exchange amongst the four users, which makes the throughput performance of the scheme better than the other existing schemes. This scheme that utilizes two channel uses for the entire information exchange between three and four users using a relay in \cite{SVR} and \cite{ShR} respectively, is extended to $n$ users in this paper, for the case when $M$-PSK is used at the end nodes.

For our physical layer network coding strategy we use the mathematical structure called a Latin Hyper-Cube defined as follows:
\begin{definition}
An \textit{n-fold Latin Hyper-Cube L of r-th order of side $M$} \cite{Kis} is an $M \times M \times ... \times M ~(n~times) $ array containing $M^n$ entries, $M^{n-r}$ of each of $M^r$ kinds, such that each symbol occurs at most once for each value taken by each dimension of the hyper-cube. \footnote{The definition has been modified slightly from the referred article ``On Latin and Hyper-Graeco-Latin Cubes and Hyper Cubes'' by K. Kishen (Current Science, Vol. 11, pp. 98--99, 1942), in accordance with the context.}
\end{definition}

For our purposes, we use only $n$-fold Latin Hyper-Cubes of side $M$ on the symbols from the set $\mathbb{Z}_{t}=\left\{0,1,2,...,t-1\right\}$, $t \geq M^{n-1}$. 
\vspace{-0.15 cm}

\section{Signal Model}
\noindent \textbf{\textit{Multiple Access (MA) Phase:}}\\
\indent Suppose user node $X_k$ for $k=1,2,...,n$ wants to send a $\lambda $-bit binary tuple to all user nodes $X_l$ for $l=1,2,...,n$ and $l \neq k$. The symmetric $M$-PSK $(M=2^ \lambda)$ constellation, given by $\mathcal{S}=\left\{e^{2 \pi k /M} | k=0,1,...,M-1 \right\}$ is used at $X_1$, $X_2$,..., $X_n$, and $ \mu : \mathbb{F}^{\lambda}_{2} \rightarrow \mathcal{S} $ denotes the map from bits to complex symbols used at $X_1$, $X_2$,..., $X_n$ where $\mathbb{F}_{2}=\left\{0,1\right\}$. Let $ x_{1}=\mu\left(s_{1}\right), x_{2}=\mu\left(s_{2}\right),..., x_{n}=\mu\left(s_{n}\right) \in \mathcal{S}$ denote the complex symbols transmitted by $X_1$, $X_2$,..., $X_n$ respectively, where $s_{1}, s_{2}, ..., s_n \in \mathbb{F}^{\lambda}_{2}$. Here, we assume that the Channel State Information (CSI) is not available at the transmitting nodes and perfect CSI is available at the receiving nodes. The received signal at R in the MA phase is given by,
{\vspace{-.25 cm}
\begin{equation}
\label{yr}
Y_{R}=H_{1}x_{1}+H_{2}x_{2}+...+H_{n}x_{n}+Z_{R},
\vspace{-0.25 cm}\end{equation}}where $H_{1}$, $H_{2}$,..., $H_{n}$ are the fading coefficients associated with the $X_1$-R, $X_2$-R,..., $X_n$-R link respectively. The additive noise $Z_{R} \sim \mathcal{CN}\left(0,\sigma^2 \right)$, where $\mathcal{CN}\left(0,\sigma^2 \right)$ denotes the circularly symmetric complex Gaussian random variable with variance $\sigma^2$.

%
The effective constellation seen at the relay during the MA phase, denoted by $ \mathcal{S}_{R} \left( H_1, H_2,..., H_n \right)$, is given by,\\ {\footnotesize $\mathcal{S}_{R} \left( H_1, H_2,..., H_n \right) = \left\{H_1 x_{1} + H_2 x_{2} +...+ H_n x_{n}| x_{1}, x_{2},..., x_n \in \mathcal{S}\right\}.$}

The minimum distance between the points in the constellation $ \mathcal{S}_{R} \left( H_1, H_2, ..., H_n \right) $ denoted by $d_{min}\left(H_1, H_2, ..., H_n\right)$ is given in (\ref{dist}) on the next page. From (\ref{dist}), it is clear that there exists values of $(H_1, H_2, ..., H_n)$, for which $d_{min}\left(H_1, H_2,..., H_n\right)=0$.

\begin{definition}
A fade state $(H_1,H_2,..., H_n)$ is defined to be a \textit{singular fade state} for the MA phase of $n$-way relaying, if $d_{min}\left(H_1, H_2,..., H_n\right)=0$. Let $\mathcal{H}=\left\{ (H_1, H_2,..., H_n) \in \mathbb{C}^n | d_{min}\left(H_1, H_2, ..., H_n\right)=0 \right\}$ denote the set of all singular fade states. For singular fade states, $\left|\mathcal{S}_{R} \left( H_1, H_2, ..., H_n \right)\right| < M^{n}$.
\end{definition}

Let the Maximum Likelihood (ML) estimate of $\left(x_{1}, x_{2},..., x_n\right) $ be denoted by $\left(\hat{x}_{1}, \hat{x}_{2},..., \hat{x}_n\right) \in \mathcal{S}^{n}$ at R based on the received complex number $Y_{R}$, i.e., 

{\vspace{-.25 cm}
\footnotesize
\begin{equation}
\left(\hat{x}_{1}, \hat{x}_{2},..., \hat{x}_n\right)=\arg \min_{\left({x_{1}}, {x_{2}}, ..., x_n\right) \in \mathcal{S}^{n}}\left\|Y_R - HX\right\|, 
\vspace{-.25 cm} \end{equation}}where $ H=\left[H_{1}\  H_{2}\  ... \ H_n \right]$  and $ X=\left[ x_{1}\ x_{2}\ ...\ x_n \right]^T.$

\begin{figure*}
\footnotesize
\begin{align}
&
\label{dist}
d_{min}(H_1, H_2, ..., H_n)=\hspace{-0.5 cm}\min_{\substack {{(x_1,x_2,...,x_n),(x'_1,x'_2,...,x'_n) \in \mathcal{S}^{n}} \\ {(x_1,x_2,...,x_n) \neq (x'_1,x'_2,...,x'_n)}}}\hspace{-0.5 cm}\vert H_1 \left(x_1-x'_1\right)+H_2 \left(x_2-x'_2\right) + ... + H_n \left(x_n-x'_n\right)\vert \\
\hline
\vspace{1cm}
&
\label{cl1}
d_{min}^{\mathcal{L}_{i},\mathcal{L}_{j}}\left(H_1, H_2,..., H_n\right)=\hspace{-0.2 cm}\min_{\substack {{(x_1,x_2,...,x_n) \in \mathcal{L}_{i}},\\ (x'_1,x'_2,...,x'_n) \in \mathcal{L}_{j}}} \hspace{-0.2 cm}  \left| H_1 \left( x_1-x'_1\right)+ H_2 \left(x_2-x'_2\right) + ... + H_n \left(x_n-x'_n\right) \right| \\
\hline
\vspace{1cm}
&
\label{cl2}
d_{min} \left(\mathcal{C}^{H_1, H_2, ..., H_n}\right)=\hspace{-1.7 cm}\min_{\substack {{(x_1,x_2,...,x_n),(x'_1,x'_2,..., x'_n) \in \mathcal{S}^{n},} \\ {\mathcal{M}^{H_1, H_2, ..., H_n}(x_1,x_2,...,x_n) \neq \mathcal{M}^{H_1, H_2, ...,H_n}(x'_1,x'_2,...,x'_n)}}}\hspace{-1.7 cm} \left| H_1 \left( x_1-x'_1\right)+ H_2 \left(x_2-x'_2\right) + ... + H_n \left(x_n-x'_n\right) \right| \\
\hline
\vspace{0.cm}
&
\label{cl3}
d_{min} \left(\mathcal{C}^{\left\{\left(H_1, H_2, ...,H_n\right)\right\}} , H_1, H_2, ..., H_n \right)=\hspace{-2.9 cm}\min_{\substack {\vspace{0.15cm} {(x_1,x_2,...,x_n),(x'_1,x'_2,...,x'_n) \in \mathcal{S}^{4},} \\ {\mathcal{M}^{H_1,H_2,...,H_n}(x_1,x_2,...,x_n) \neq \mathcal{M}^{H_1, H_2, ...,H_n}(x'_1,x'_2,...,x'_n)}}}\hspace{-2.6 cm}  \left| H_1 \left( x_1-x'_1\right)+ H_2 \left(x_2-x'_2\right) + ... + H_n \left(x_n-x'_n\right)\right|\\
\hline
\vspace{0.cm}
&
\label{mel}
\mathcal{M}_{k}^{H_1,H_2,..., H_k,...,H_n}\left(x_{1},x_{2},...,x_k,...,x_n\right) \neq \mathcal{M}_{k}^{H_1,H_2,...,H_k,...H_n}\left(x'_{1},x'_{2},...,x'_k,...,x'_n\right),\\
\vspace{0.cm}
\nonumber
&
where, ~x_k = x'_k, \left(x_{1},x_{2},...,x_{k-1},x_{k+1},...,x_n\right) \neq \left(x'_{1},x'_{2},...,x'_{k-1},x'_{k+1},...,x'_n\right), \forall x_{1}, x'_{1},x_{2}, x'_{2},..., x_{n}, x'_{n}  \in \mathcal{S} ~for ~k=1,2,...,n.\\
\hline
\nonumber
\end{align}
\vspace{-1.4cm}
\end{figure*}

\noindent \textbf{\textit{Broadcast (BC) Phase:}}\\
\indent During the BC phase, the received signals at $X_1,~ X_2, ..., ~X_n$ are respectively given by, 
{\vspace{-.2 cm}
\begin{equation}
 Y_{X_k}=H'_{k}X_{R}+Z_{k}, ~k=1,2,...,n;
\vspace{-.5 cm} \end{equation}}

\noindent where $X_{R}=\mathcal{M}^{H_1, H_2, ..., H_n}\left(\left(\hat{x}_{1},\hat{x}_{2},...,\hat{x}_{n}\right)\right) \in \mathcal{S}^{'}$ denotes the complex number transmitted by R and $H_{1}^{'},$ $H_{2}^{'},$..., $H_{n}^{'}$ respectively are the fading coefficients corresponding to the links R-$X_1$, R-$X_2$, ..., R-$X_n$. The additive noises $Z_{1},$ $Z_{2},$,..., $Z_{n}$ are  $\mathcal{CN}\left(0,\sigma^{2}\right)$. During the \textit{BC} phase, R transmits a point from a signal set $\mathcal{S}^{'}$ given by a many to one map $\mathcal{M}^{H_1, H_2, ..., H_n} : \mathcal{S}^n \rightarrow \mathcal{S}^{'} $ chosen by R, depending on the values of $H_1$, $H_2$, ..., $H_n$. The cardinality of $\mathcal{S}^{'} \geq 2^{\lambda(n-1)}$, since $\lambda (n-1)$ bits about the other $(n-1)$ users needs to be conveyed to each of $X_1$, $X_2$,..., $X_n$.

A \textit{cluster} is the set of elements in $\mathcal{S}^n $ which are mapped to the same signal point in $\mathcal{S}^{'}$ by the map $\mathcal{M}^{H_1, H_2, ..., H_n}$. Let $\mathcal{C}^{H_1, H_2, ..., H_n}=\left\{\mathcal{L}_{1}, \mathcal{L}_{2},..., \mathcal{L}_{l}\right\}$ denote the set of all such clusters.

\begin{definition}
The \textit{cluster distance} between a pair of clusters $\mathcal{L}_i, ~\mathcal{L}_j \in \mathcal{C}^{H_1, H_2, ..., H_n}$, as given in (\ref{cl1}) on the next page, is the minimum among all the distances calculated between the points $\left(x_{1}, x_{2}, ..., x_n\right) \in \mathcal{L}_{i}$ and $\left(\acute{x_{1}}, \acute{x_{2}},..., \acute{x_{n}}\right) \in \mathcal{L}_{j}$ in the effective constellation seen at the relay node R. The minimum among all the cluster distances among all pairs of clusters of a clustering $\mathcal{C}^{H_1, H_2,..., H_n}$ is the \textit{minimum cluster distance} of the clustering, as given in (\ref{cl2}) on the next page.
\end{definition}

During the MA phase, the performance depends on the minimum cluster distance, while during the BC phase, the performance is dependent on the minimum distance of the signal set $\mathcal{S}^{'}$. \textit{Distance shortening}, a phenomenon given in \cite{NMR}, is described as the significant reduction in the value of $d_{min}\left(H_1, H_2, ..., H_n \right)$ for values of $\left(H_1, H_2,..., H_n\right)$ in the neighborhood of the singular fade states. If the clustering used at the relay node R in the BC phase is chosen such that $d_{min}(\mathcal{C}^{H_1, H_2, ..., H_n})$ is non zero, then the effect of distance shortening can be avoided.

A clustering $\mathcal{C}^{H_1, H_2, ..., H_n}$ is said to remove a singular fade state $\left(H_1, H_2, ..., H_n\right) \in \mathcal{H}, $ if $d_{min}\left(\mathcal{C}^{H_1, H_2, ..., H_n}\right)>0$, i.e., any two message sequences $\left(x_{1},x_{2},...,x_n\right) \in \mathcal{S}^n$ that coincide in the effective constellation received at the relay during the MA phase is in the same cluster of $\mathcal{C}^{H_1, H_2, ..., H_n}$. So, removing singular fade states for a $n$-way relay channel can alternatively be defined as:
\begin{definition}
A clustering $\mathcal{C}^{H_1, H_2, ..., H_n}$ is said to \textit{remove the singular fade state} $\left(H_1, H_2, ..., H_n\right) \in \mathcal{H}$, if any two possibilities of the messages sent by the users $\left(x_{1},x_{2},...,x_n\right), \left(x'_{1},x'_{2},...,x'_n\right) \in \mathcal{S}^{n}$ that satisfy

{\footnotesize  $$ H_1 x_1+ H_2 x_2+ ...+H_n x_n=H_1 x'_1+ H_2 x'_2+ ... +H_n x'_n $$} are placed together in the same cluster by the clustering.
\end{definition}

We denote the clustering which removes the singular fade state $\left(H_1, H_2, ..., H_n\right)$  by $\mathcal{C}^{\left\{\left(H_1, H_2, ..., H_n\right)\right\}} $ (selecting one randomly if there are multiple clusterings which remove the same singular fade state $\left(H_1, H_2, ..., H_n\right)$). Let the set of all such clusterings be denoted by $\mathcal{C_{H}}$, i.e., $\mathcal{C_{H}}=\left\{\mathcal{C}^{\left\{\left(H_1, H_2, ..., H_n\right)\right\}} : \left(H_1, H_2, ..., H_n\right)\in \mathcal{H}\right\} $.

\begin{definition}
The minimum cluster distance of the clustering $\mathcal{C}^{\left\{\left(H_1, H_2, ..., H_n\right)\right\}}$ for $\left(H_1, H_2, ...,H_n\right) \in \mathcal{H}$, when the fade state $(H_1, H_2, ..., H_n)$ occurs in the MA phase, denoted by $d_{min}\left(\mathcal{C}^{\left\{\left(H_1, H_2, ..., H_n\right)\right\}},H_1, H_2, ..., H_n\right)$, is the minimum among all its cluster distances.
\end{definition}

If $\left(H_1, H_2, ..., H_n\right) \notin \mathcal{H}, $ the clustering $\mathcal{C}^{H_1, H_2, ..., H_n} $ is chosen to be $\mathcal{C}^{\left\{\left(H_1, H_1, ..., H_n\right)\right\}} \in \mathcal{C}_{\mathcal{H}}$, that satisfies, \\
{$ d_{min}\left(\mathcal{C}^{\left\{\left(H_1, H_2, ..., H_n\right)\right\}}, H_1, H_2, ..., H_n\right) \geq $\\
$~~~~~~~~~~~~~~~~~~~~~~~~~~~~d_{min} \left(\mathcal{C}^{\left\{\left(H'_1, H'_2, ..., H'_n\right)\right\}},H_1, H_2, ..., H_n\right),$\\
$\forall \left(H_{1}, H_{2}, ..., H_n\right) \neq \left(H'_{1}, H'_{2}, ..., H'_n\right) \in \mathcal{H}$}.
The clustering used by the relay is indicated to $X_1$, $X_2$,...,$X_n$ using overhead bits. In order to ensure that $X_k;~k=1,2,..,n$ is able to decode the message sent by $X_l;~l=1,2,..,n;~l \neq k$, the clustering $\mathcal{C}$ should satisfy the exclusive law, as given in (\ref{mel}). We explain Exclusive Law in more detail in the next section.

%
The contributions of this paper are as follows: 
\begin{itemize}
\item We propose a scheme that enables the exchange of information in the wireless $n$-way relaying scenario when $M$-PSK is used at the $n$ user nodes with totally two channel uses while attempting to remove the harmful effects of fading, extending the schemes given in \cite{KPT}, \cite{SVR}, \cite{ShR} for $n=2,3,4$ respectively. 
\item For this scheme, the singular fade spaces are identified, enumerated and classified based on their removability in the given scenario. 
\end{itemize}
The remaining content is organized as follows: Section III demonstrates how a $n$-fold Latin Hyper-Cube of side $M$ can be utilized to represent the network code that satisfies the exclusive law for $n$-way relaying when $M$-PSK is used at the end nodes. In Section IV we describe and enumerate the singular fade subspaces for the given scenario and in Section V, focus in on the removal of such singular fade subspaces using $n$-fold Latin Hyper-Cube of side $M$. Section VI provides some insights using simulations and Section VII concludes the paper. 
%
\section{The Exclusive Law and Latin Hyper-Cubes}

The clustering $\mathcal{C}$ that represents the map used at the relay should satisfy the exclusive law \cite{KPT} in order to ensure that $X_k;~k=1,2,...,n$ is able to decode the message sent by $X_l;~l=1,2,...,n;~l \neq k$, where we assume that the nodes $X_1,~X_2,...,~X_n$ transmit symbols from the $M$-PSK constellation. Consider a $M \times M \times ... \times M ~(n~times) $ array, containing $M^n$ entries indexed by $\left(x_{1}, x_{2}, ..., x_n\right)$, i.e., the $n$ symbols sent by $X_1,~X_2,...,~X_n$ in the MA phase. For $k=1,2,...,n$, fixing the $k^{th}$ dimension of this $M \times M \times ... \times M ~(n~times) $ array, the $M$  $(n-1)$ dimensional arrays obtained, denoted by say $ \mathcal{C}_{k}^{l}, ~l=1,2,...,M$, are indexed by the $M$ values taken by $x_{k}$. For fixed values of $k$ and $l$, the repetition of a symbol in $\mathcal{C}_{k}^{l}$ results in the failure of the $k^{th}$ exclusive law given by (\ref{mel}). Thus, for the exclusive law to be satisfied, the cells of this array should be filled such that the $M \times M \times ... \times M ~(n~times) $ array so obtained, is a $n$-fold Latin Hyper-Cube of side $M$, with entries from $\mathbb{Z}_t=\left\{ 0,1,...,t-1 \right\}$ for $t \geq M^{(n-1)}$ (Definition 1). The symbol $\mathcal{L}_i$ of a particular clustering $\left\{\mathcal{L}_1,...,\mathcal{L}_t\right\}$ denotes the cluster obtained by putting together all the tuples $\left(x_1,x_2,...,x_n\right)\in \mathcal{S}^n $ such that the entry in the $\left(x_1,x_2,...,x_n\right)$-th slot is the same entry $i$ from $\mathbb{Z}_{t}$. The adjoining figures Fig. \ref{fig:three} and Fig. \ref{fig:four} show the exclusive law condition for the three-way and four-way relaying scenario when 4-PSK is used at end nodes.

\begin{figure}[ht]
\center
\vspace{-.2cm}
\includegraphics[height=25mm]{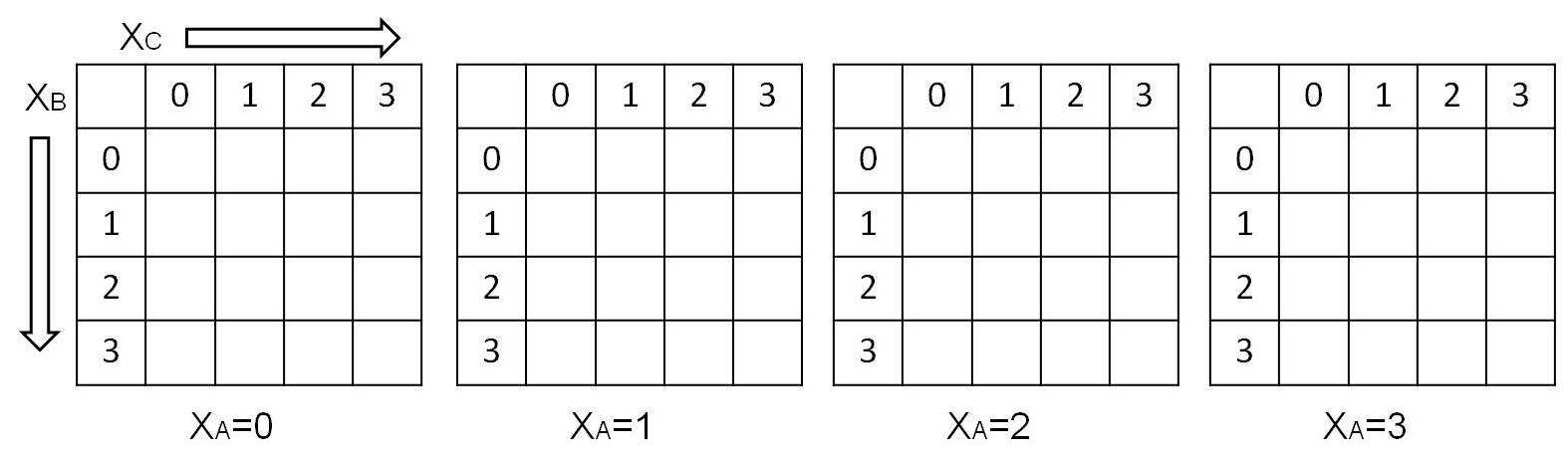}
\vspace{-.4cm}
\caption{A 4-fold Latin Hyper-Cube of side 3 represents the exclusive law constraint for the relay map when 4-PSK is used at end nodes}
\label{fig:three}
\vspace{-.5cm}
\end{figure}

\begin{figure}[ht]
\center
\vspace{-.2cm}
\includegraphics[height=32mm]{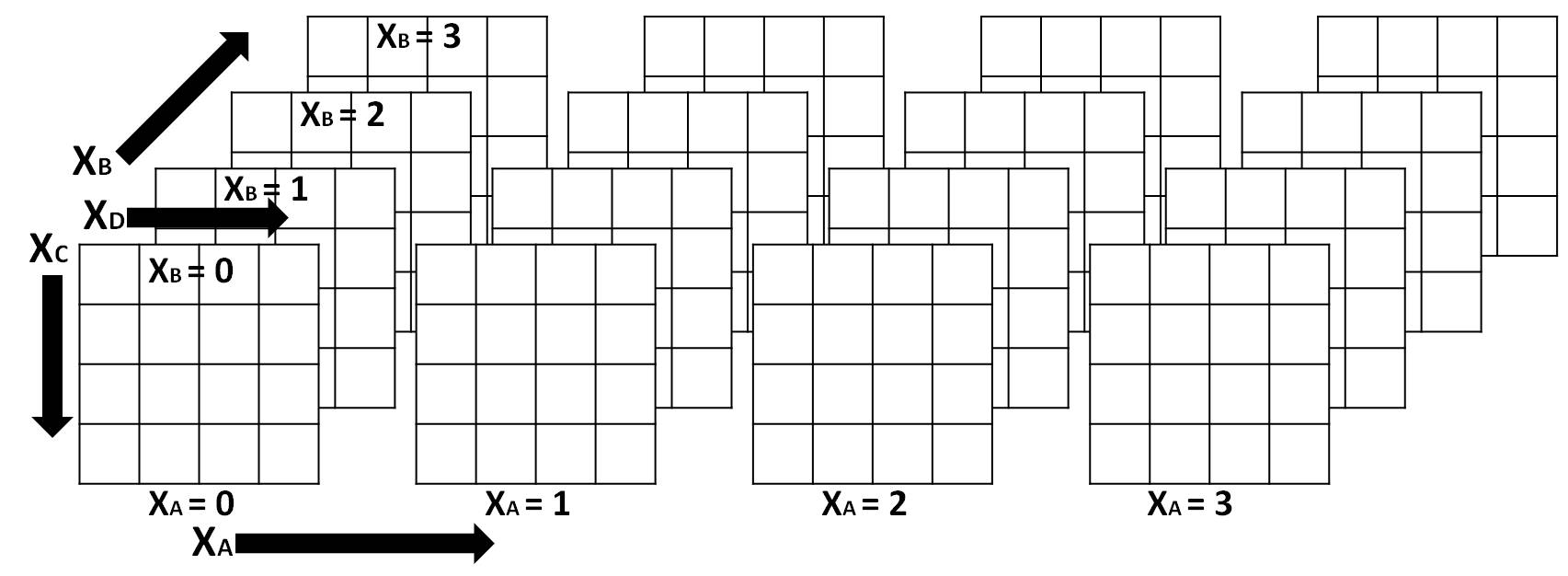}
\vspace{-.4cm}
\caption{A 4-fold Latin Hyper-Cube of side 4 represents the exclusive law constraint for the relay map when 4-PSK is used at end nodes}
\label{fig:four}
\vspace{-.5cm}
\end{figure}
\vspace{-.cm}

\section{Singular Fade Subspaces}

\begin{definition}
A set $\left\{(x_1, x_2,...,x_n)\right\} \in \mathcal{S}^n $ consisting of all the possibilities of $(x_1, x_2, ..., x_n)$ that must be placed in the same cluster of the clustering used at relay node R in the BC phase in order to remove the singular fade state $\left(H_1, H_2, ..., H_n\right)$, is referred to as a \textit{Singularity Removal Constraint} for the singular fade state $\left(H_1, H_2, ..., H_n\right)$ for $n$-way relaying scenario.
\end{definition}

At the end of the MA phase, the relay node receives a complex number, given by (\ref{yr}). Using the ML estimate of this received complex number, R transmits a point from the constellation $\mathcal{S}'$ with cardinality at most $M^{n}$. Instead of R transmitting a point from the $M^{n}$ point constellation resulting from all the possibilities of $\left(x_{1}, x_{2}, ..., x_n\right)$, depending on the fade states, the relay R can choose to group these possibilities into clusters represented by a smaller constellation, so that the minimum cluster distance is non-zero, as well as all the users receive the messages from the other $(n-1)$ users, i.e., the clustering satisfies the exclusive law. We provide one such clustering for the case of $n$-way relaying in the following. 

Suppose the fade coefficient in the MA phase, $(H_1, H_2,..., H_n) $, is a singular fade state, and ${\Gamma}$ is a singularity removal constraint corresponding to the singular fade state $(H_1, H_2, ..., H_n) $. Then there exist $(x_1, x_2, ..., x_n), (x'_1, x'_2, ...,x'_n) \in \Gamma $, $(x_1, x_2, ..., x_n) \neq (x'_1, x'_2, ..., x'_n) $ such that, 

{\footnotesize
\vspace{-.2 cm}
\begin{align}
\nonumber
& H_1 x_1+ H_2 x_2+ ...+ H_n x_n=H_1 x'_1+ H_2 x'_2+... + H_n x'_n \\
\nonumber
\Rightarrow & H_1 (x_1-x'_1)+ H_2 (x_2-x'_2)+ ...+ H_n (x_n-x'_n)=0 \\
\Rightarrow & (H_1, H_2, ..., H_n) \in 
\left\langle  \left[ {\begin{array}{cc}
x_{1}-x'_1 \\
x_{2}-x'_2 \\
.. \\
x_{n}-x'_n \\
\end{array} } \right] \right\rangle ^{\bot} 
\label{nullspacen} 
\vspace{-.2 cm}
\end{align}}where for a $n \times 1$ non-zero vector $v$ over $\mathbb{C}$,

{
\footnotesize
\vspace{-0.6 cm}
\begin{equation}
\left\langle v\right\rangle ^{\bot}=\left\{w=(w_1,w_2,...,w_n) \in \mathbb{C}^n ~|~ w_1v_1+w_2v_2+...+w_n v_n=0 \right\}.
\vspace{-. cm}
\end{equation}}

Note that $w_1v_1+w_2v_2+...+w_n v_n= \sum_{\substack{i}} w_i v_i $ is the dot product over $\mathbb{C}$ (and not an inner product over $\mathbb{C}$).
\begin{lemma} For a $n \times 1$ non-zero vector $v$ over $\mathbb{C}$, $\left\langle v\right\rangle ^{\bot}$ is a $(n-1)$-dimensional vector subspace of $\mathbb{C}^n$ over $\mathbb{C}$. \footnote{The proof is straightforward, yet given here for the sake of completeness.}
\end{lemma}
\begin{proof} Let $w=(w_1,w_2,...,w_n) \in \left\langle v\right\rangle ^{\bot}$ where $v=(v_1,v_2,...,v_n)$. Then, by definition, 

\begin{align}
\nonumber
& w_1v_1+w_2v_2+...+w_n v_n=0 \\
\nonumber
\Rightarrow & w_1v_1=-(w_2v_2+...+w_n v_n) \\
\nonumber
\Rightarrow & w_1=-v^{-1}_{1}(w_2v_2+...+w_n v_n) \\
\nonumber
\Rightarrow & (w_1,w_2,...,w_n) =  (-v^{-1}_{1}(w_2v_2+...+w_n v_n), w_2, ...,w_n) \\
\nonumber
\Rightarrow & (w_1,w_2,...,w_n) =  w_2 (-v^{-1}_{1} v_2, 1,0, ...,0)+\\
\nonumber
&~~~~ w_3 (-v^{-1}_{1}v_3, 0,1, ...,0)+...+ w_n (-v^{-1}_{1}v_n, 0,0, ...,1) \\
\nonumber
\Rightarrow & \left\langle v\right\rangle ^{\bot} =  span \left\{(-v^{-1}_{1} v_2, 1,0, ...,0),(-v^{-1}_{1} v_3, 0,1, ...,0),... \right. \\
\nonumber
& \left. ~~~~~~~~~~~~~~~~~~~~~~~~~~~ ..., (-v^{-1}_{1} v_n, 0,0,...,1)\right\} \text{~over~} \mathbb{C}. \\
\nonumber
\vspace{-1cm}
\end{align}Proving that the following subset of $\mathbb{C}^n$\\
{\footnotesize $ \left\{(-v^{-1}_{1} v_2, 1,0, ...,0),(-v^{-1}_{1} v_3, 0,1, ...,0), ..., (-v^{-1}_{1} v_n, 0,0,...,1)\right\}$} \\
is a linearly independent set over $\mathbb{C}$ of cardinality $(n-1)$ would be sufficient to prove that $\left\langle v\right\rangle ^{\bot}$ is a $(n-1)$-dimensional vector subspace of $\mathbb{C}^n$. Let $\alpha_1,~\alpha_2,...,~\alpha_{n-1} \in \mathbb{C}$ such that 
{ \footnotesize
\begin{align}
\nonumber
& \alpha_1 (-v^{-1}_{1} v_2, 1,0, ...,0) + \alpha_2 (-v^{-1}_{1}v_3, 0,1, ...,0)+ ...\\
\nonumber
& ~~~~~~~~~~~~~~~~~~~~~~~~~~~~...+ \alpha_{n-1} (-v^{-1}_{1}v_n, 0,0, ...,1) = (0,~0,...,~0) \\
\nonumber
\Rightarrow & (-\alpha_1 v^{-1}_{1} v_2 - \alpha_2 v^{-1}_{1}v_3 -...- \alpha_n v^{-1}_{1}v_n,~ \alpha_1,~ \alpha_2,..., ~ \alpha_{n-1}) \\
\nonumber
&~~~~~~~~~~~~~~~~~~~~~~~~~~~~~~~~~~~~~~~~~~~~~~~~~~~~~~~~~~~~~~~~~~ =(0,~0,...,~0). \\
\nonumber
\vspace{-1cm}
\end{align}}Comparing the $2^{nd},3^{rd},..., n^{th}$ components of the LHS and RHS, we get, $\alpha_1=0,\alpha_2=0, ..., \alpha_{n-1}=0$.

\end{proof}

Since $x_1, x_2, ..., x_n, x'_1, x'_2,..., x'_n \in \mathcal{S}$, where $\mathcal{S}$ is finite, there are only finitely many possibilities for the right-hand side of (\ref{nullspacen}). Thus the uncountably infinite singular fade states $(H_1, H_2, ..., H_n)$, are points in a finite number of $(n-1)$-dimensional vector subspaces of $\mathbb{C}^n$ over $\mathbb{C}$. We shall refer to these finite number of vector subspaces as the \textit{Singular Fade Subspaces} \cite{SVR}. 

We now give a detailed description of all the possibilities of singular fade subspaces for $n$-way relaying scenario when $M$-PSK is used at the end nodes. For the proof of the following Theorem, we extensively use the following Lemmas given with proofs in \cite{MuR}.

\begin{lemma} Let $\Delta \mathcal{S}$ denote the difference constellation of the signal set $\mathcal{S}$, i.e., $\Delta \mathcal{S}= \left\{s_i-s'_i | s_i,s'_i \in \mathcal{S}\right\}$. Then, for any $M$-PSK signal set, $\Delta \mathcal{S}$ is of the form, 
\begin{align}
\Delta \mathcal{S} =\left\{ 0 \right\} & \cup \left\{ 2 \sin( \frac{ \pi l}{M}) e^{j \frac{2 \pi k}{M}} | ~l ~is ~odd\right\}\\
\nonumber
&\cup \left\{ 2 \sin( \frac{ \pi l}{M}) e^{j( \frac{2 \pi k}{M} + \frac{\pi}{M})} | ~l ~is ~even\right\},
\end{align}
where $ 1 \leq l \leq M/2 $ and $ 0 \leq k \leq M-1$.
\end{lemma}

As a result of the above Lemma, the non-zero points in $\Delta \mathcal{S}$ lie on $M/2$ circles of radius $2 \sin( \pi l/M), 1 \leq l \leq M/2$ with each circle containing $M$ points. The phase angles of the $M$ points on each circle is $2 k \pi /M$, if $l$ is odd and $2 k \pi /M + \pi/M $ if $l$ is even, where $0 \leq k \leq M-1$.

\vspace{.2cm}
\begin{lemma}\cite{SVR} Let $i_1, i_2,...,i_L$ be the ordered indices corresponding to the non-zero components in $ \Delta x$ and $\Delta x'$ (the location of non-zero components is the same in the vectors $ \Delta x$ and $\Delta x'$). For $M$-PSK signal set, $ |\Delta x_i|=c|\Delta x'_i|, \forall 1 \leq i \leq n$, for some $c \in \mathbb{C}$, if and only if the magnitudes of the non-zero components in $\Delta x$ are equal and the magnitudes of the non-zero components in $\Delta x'$ are equal, i.e., $ |\Delta x_{i_1}|=|\Delta x_{i_2}|=...=|\Delta x_{i_L}|$ and $ |\Delta x'_{i_1}|=|\Delta x'_{i_2}|=...=|\Delta x'_{i_L}|$.
\end{lemma}

From (\ref{nullspacen}) and Lemma 2, the singular fade subspaces are given by, \vspace{-0.4cm}

{\footnotesize $$ \left\langle  \left[ {\begin{array}{cc} x_{1}-x'_1 \\ x_{2}-x'_2 \\ . \\ . \\ x_{n}-x'_n \\ \end{array} } \right] \right\rangle ^{\bot} = \left\langle  \left[ {\begin{array}{cc} 2 \sin( \frac{ \pi l_1}{M}) e^{j m_1} \\ 2 \sin( \frac{ \pi l_2}{M}) e^{j m_2} \\ . \\. \\ 2 \sin( \frac{ \pi l_n}{M}) e^{j m_n} \\ \end{array} } \right] \right\rangle ^{\bot} $$ \vspace{-0.4cm}
$$ ~~~~~~~~~~~~~~~~~~~~~~~~~~~~= \left\langle  \left[ {\begin{array}{cc}  \sin( \frac{ \pi l_1}{M}) e^{j m_1} \\  \sin( \frac{ \pi l_2}{M}) e^{j m_2} \\ . \\. \\  \sin( \frac{ \pi l_n}{M}) e^{j m_n} \\ \end{array} } \right] \right\rangle ^{\bot} $$}where $m_i=2 k_i \pi /M$ if $l_i$ is odd and $2 k_i \pi /M + \pi/M $ if $l_i$ is even, where $0 \leq k_i \leq M-1$ for $i=1,2,...,n$. 

\begin{theorem} There are $\sum^{n}_{k=1} (^{n}_{k}) \left[ (\frac{M}{2})^k - (\frac{M}{2}) +1\right] M^{k-1}$ Singular Fade Subspaces for $n$-way relaying when $M$-PSK constellation is used at the end nodes.
\end{theorem}
\begin{proof} The Singular Fade Subspaces are of the form $\left\langle \left[\Delta x_1, \Delta x_2,..., \Delta x_n\right]\right\rangle^{\bot} $ where $ \Delta x_k \in \Delta \mathcal{S}, ~k=1,2,...,n$. Let $k$ be the number of non-zero $x'_{i}s$. We fix the relative phase vector of the vector $\left[\Delta x_1, \Delta x_2,..., \Delta x_n\right] =w ~(say)$. The points in $\Delta \mathcal{S} $ lie on $M/2$ circles. So there are $ (M/2)^k$ possibilities for absolute values of the non-zero components of $w$. There are $M/2$ possibilities for the case that the absolute values of all the components of $w$ that are non-zero, are equal. From Lemma 3, the Singular Fade Subspaces resulting from all of these $M/2$ cases are the same, and hence account for 1 out of the $ (M/2)^k$ cases. So for a fixed relative phase vector, there are $ \left[ (M/2)^k - M/2 +1 \right]$ Singular Fade Subspaces. From Lemma 3, fixing the absolute values of the non-zero components of $w$, each distinct relative phase vector corresponds to a distinct Singular Fade Subspace. There are $M^{k-1}$ distinct possibilities for the relative phase vector. So, there are $ \left[ (M/2)^k - M/2 +1 \right]M^{k-1}$ Singular Fade Subspaces when $w$ has $k$ non-zero components. Here, $k$ can take values from 1 to $n$. Summing over all possible values of $k$, we have $\sum^{n}_{k=1} (^{n}_{k}) \left[ (\frac{M}{2})^k - (\frac{M}{2}) +1\right] M^{n-1}$ Singular Fade Subspaces for $n$-way relaying when $M$-PSK constellation is used at the end nodes.
\end{proof}

The above theorem coincides with the results given for $n=2$ in \cite{MuR}, and the results obtained using explicit enumeration for $n=3$ in \cite{SVR} and for $n=4$ in \cite{ShR}. For illustration, we discuss the case when $n=5$. In five-way relaying, user nodes (say) A, B, C, D and E transmit $x_A$, $x_B$, $x_C$, $x_D$ and $x_E$ $\in \mathcal{S}$ respectively in the first channel use. Suppose the fade coefficient in the MA phase, $(H_A, H_B, H_C, H_D, H_E) $, is a singular fade state. Then there exist $(x_A, x_B, x_C, x_D, x_E), (x'_A, x'_B, x'_C, x'_D ,x'_E) \in \mathcal{S}^5 $, $(x_A, x_B, x_C, x_D, x_E)\neq (x'_A, x'_B, x'_C, x'_D ,x'_E) $ such that, 

{\footnotesize
\vspace{-.2 cm}
\begin{align}
\nonumber
& H_A x_A+ H_B x_B+ H_C x_C+ H_D x_D+ H_E x_E \\
\nonumber
&~~~~~~~~~~~~~~~~~~~~~~~=H_A x'_A+ H_B x'_B+ H_C x'_C+ H_D x'_D+ H_E x'_E \\
\nonumber
\Rightarrow & H_A (x_A-x'_A)+ H_B (x_B-x'_B)+ H_C (x_C-x'_C) \\
\nonumber
&~~~~~~~~~~~~~~~~~~~~~~~~~~~~~~~~~~+H_D (x_D-x'_D) + H_E (x_E-x'_E)=0 \\
\Rightarrow & (H_A, H_B, H_C, H_D, H_E) \in 
\left\langle  \left[ {\begin{array}{cc}
x_{A}-x'_A \\
x_{B}-x'_B \\
x_{C}-x'_C \\
x_{D}-x'_D \\
x_{E}-x'_E \\
\end{array} } \right] \right\rangle ^{\bot}. 
\label{nullspace} 
\vspace{-.2 cm}
\end{align}}

The adaptive network coding for five-way relaying attempts at removing the singular fade subspaces for the case given by, {\footnotesize $ \left\langle  \left[ {\begin{array}{cc}
x_{A}-x'_A \\
x_{B}-x'_B \\
x_{C}-x'_C  \\
x_{D}-x'_D  \\
x_{E}-x'_E  \\
\end{array} } \right] \right\rangle ^{\bot} 
$}. In the second channel use, relay node R transmits $x_R$ using a network coding map that depends on the values of $\hat{x_A}$, $\hat{x_B}$, $\hat{x_C}$, $\hat{x_D}$ and $\hat{x_E}$. As explained in Section III, using a network coding map represented by a 5-fold hyper latin-cube of side $M$ (when $M$-PSK is used at the end nodes A, B, C, D and E) ensures that exclusive law is satisfied. It can be shown using explicit enumeration, that when the end nodes use 4-PSK, there are $13981$ singular fade subspaces for five-way relaying, which coincides with Theorem 1, for $n=5,~ M=4$. 

\section{Removing singular fade subspaces}
\vspace{-.1cm}
We cluster the possibilities of $\left(x_{1},x_{2},...,x_n\right)$ into a clustering that can be represented by an $n$-fold Latin Hyper-Cubes of side $M$, to obtain a clustering that removes the singular fade subspaces, and also attempts to minimize the size of the constellation used by R. This clustering is represented by a constellation given by $\mathcal{S}'$, which is utilized by the relay node R in the BC phase. This is done by first constraining the $M^n$ possibilities of $\left(x_{1},x_{2},...,x_n\right)$ transmitted at the MA phase, to remove the singular fade subspaces, and then using these constraints, filling the entries of an empty $M \times M \times ... \times M ~(n~times)$ array representing the map to be used at the relay. This partially filled array is completed so as to form a $n$-fold Latin hyper-cube of side $M$. The mapping to be used at R can be obtained from the complete Latin hyper-cube keeping in mind the equivalence between the relay map that satisfies the exclusive law with the $n$-fold Latin Hyper-Cube of side $M$ as shown in Section III.

\begin{algorithm}[ht]
\SetLine
\linesnumbered
\KwIn{The constrained $M \times M\times ... \times M~ (n-times)$ array}
\KwOut{A n-fold Latin Hyper-Cube of side M representing the clustering map at the relay}

Start with the constrained $M \times M\times ... \times M~ (n-times)$ array $\mathcal{X}$

Initialize all empty cells of $\mathcal{X}$ to 0

%
The $\left(i_1,i_2,...,i_n\right)^{th}$ cell of $\mathcal{X}$ is the $i_1^{th}$ transmission of $X_1$, the $i_2^{th}$ transmission of $X_2$, ..., the $i_n^{th}$ transmission of $X_n$.

\For{$1\leq i_1 \leq M $}{

\For{$1\leq i_2 \leq M $}{...

\For{$1\leq i_n \leq M $}{

\If{cell $\left(i_1,i_2,...,i_n\right)$ of $\mathcal{X}$ is NULL}{

Initialize c=1

\eIf{$\mathcal{L}_{c}$ does not occur in the $\left(i_1,i_2,...,i_n\right)^{th}$ cell of $\mathcal{X}$}{
replace 0 at cell $\left(i_1,i_2,...,i_n\right)$ of $\mathcal{X}$ with $\mathcal{L}_{c}$\;
}{
c=c+1\; 
}
}
}
}
}
\caption{Obtaining the n-fold Latin Hyper-Cube of side M from the constrained $M \times M \times ... \times M ~(n-times) $ array}

\end{algorithm}

During the MA phase for the $n$-way relaying scenario, nodes $X_1$, $X_2$,..., $X_n$ transmit to the relay R. Let the fade state $(H_1, H_2, ..., H_n)$ denote a point in one of the $\sum^{n}_{k=1} (^{n}_{k}) \left[ (\frac{M}{2})^k - (\frac{M}{2}) +1\right] M^{n-1}$ singular fade subspaces (Section IV). The constraints on the $M^n$ array representing the map at the relay node R during BC phase for a singular fade state, can be obtained using the vectors of differences, viz., $ \left[x_1-x'_1, ~ x_2-x'_2,...,~ x_n-x'_n\right]$ contributing to this particular singular fade state. So, if $(H_1, H_2, ..., H_n) \in \left\langle \left[x_1-x'_1, ~ x_2-x'_2,...,~ x_n-x'_n\right] \right\rangle^{\bot}$, then, for $\left(x_{1}, x_{2}, ..., x_n\right), \left(x'_{1},x'_{2},..., x'_n\right) \in \mathcal{S}^n$, $ H_1 x_{1} + H_2 x_{2} + ... +H_n x_n= H_1 x'_{1} + H_2 x'_{2} + ... +H_n x'_n$. For a clustering to remove the singular fade state $(H_1, H_2, ..., H_n)$, i.e., for the minimum distance of the clustering to be greater than 0 (Section II), the pair $\left(x_{1}, x_{2}, ..., x_n\right), \left(x'_{1},x'_2, ...,x'_n\right)$ must be kept in the same cluster. Alternatively, we can say that the entry corresponding to $\left(x_{1}, x_{2}, ..., x_n\right)$ in the $M^n$ array must be the same as the entry corresponding to $\left(x'_{1},x'_{2}, ..., x'_n\right)$. Similarly, every other such pair in $\mathcal{S}^n$ contributing to this same singular fade subspace must be kept in the same cluster. Apart from all such pairs in $\mathcal{S}^n$ being kept in the same cluster of the clustering, in order to remove this particular fade state, there are no other constraints. Consider the ordering given as follows on the entries of the constrained $M^n$ array: {\footnotesize $ (c_1, c_2,...,c_n) < (c'_1, c'_2,...,c'_n)$} if $c_i<c'_i$ where $i$ is the first component among the $n$ components, where $c_i \neq c'_i$. This constrained $M^n$ array can then be completed by simply filling the first empty cell in this order, with $\mathcal{L}_i, i \geq 1$ in the increasing order of $i$ such that the completed array is an $n$-fold Latin Hyper-Cube of side $M$ (Algorithm 1). The above clustering scheme, however, cannot be utilized to remove all the singular fade subspaces, as shown in the following lemma.

\begin{table*}[ht]
\vspace{-. cm}
\tiny
\centering
\renewcommand{\arraystretch}{1.2}
\begin{tabular}{!{\vrule width 1.pt}c!{\vrule width 0.9pt}p{0.3cm}|p{0.3cm}|p{0.3cm}|p{0.3cm}!{\vrule width 1.pt}c!{\vrule width 0.9pt}p{0.3cm}|p{0.32cm}|p{0.32cm}|p{0.32cm}!{\vrule width 1.pt}c!{\vrule width 0.9pt}p{0.32cm}|p{0.32cm}|p{0.32cm}|p{0.32cm}!{\vrule width 1.pt}c!{\vrule width 0.9pt}p{0.32cm}|p{0.32cm}|p{0.32cm}|p{0.32cm}!{\vrule width 1.pt}}\noalign{\hrule height 1.pt}
 $x_C=0$   &\multirow{2}{*}{0} & \multirow{2}{*}{1} & \multirow{2}{*}{2} & \multirow{2}{*}{3}                  & $x_C=0$ & \multirow{2}{*}{0} & \multirow{2}{*}{1} & \multirow{2}{*}{2} & \multirow{2}{*}{3}                                    & $x_C=0$ & \multirow{2}{*}{0} & \multirow{2}{*}{1} & \multirow{2}{*}{2} & \multirow{2}{*}{3}            &  $x_C=0$ & \multirow{2}{*}{0} & \multirow{2}{*}{1} & \multirow{2}{*}{2} & \multirow{2}{*}{3} \\
 $x_D=0$   &   &   &   &                                                                                       & $x_D=1$ &   &   &   &                                                                                                          & $x_D=2$ &   &   &   &                                                                                  &  $x_D=3$ &   &   &   &   \\\noalign{\hrule height 0.9pt}
   0       &    $\mathcal{L}_{4}$  &    $\mathcal{L}_{19}$    & $ \mathcal{L}_{23}$   &  $\mathcal{L}_{27}$   & 0   &    $\mathcal{L}_{75}$  &    $\mathcal{L}_{79}$    & $\mathcal{L}_{83}$   &  $\mathcal{L}_{87}$               & 0   &    $\mathcal{L}_{2}$    &    $\mathcal{L}_{139}$    & $\mathcal{L}_{143}$     &  $\mathcal{L}_{147}$  &  0       &    $\mathcal{L}_{197}$             &    $\mathcal{L}_{201}$    & $ \mathcal{L}_{205}$   &  $\mathcal{L}_{209}$ \\\hline
   1       &    $\mathcal{L}_{8}$  &    $\mathcal{L}_{20}$    & $ \mathcal{L}_{24}$   &  $\mathcal{L}_{28}$   & 1   &    $\mathcal{L}_{76}$  &    $\mathcal{L}_{80}$    & $ \mathcal{L}_{84}$    &  $\mathcal{L}_{88}$              & 1   &   $\mathcal{L}_{6}$    &    $\mathcal{L}_{140}$    & $ \mathcal{L}_{144}$    &  $\mathcal{L}_{148}$  &  1       &    $\mathcal{L}_{198}$             &    $\mathcal{L}_{202}$    & $ \mathcal{L}_{206}$   &  $\mathcal{L}_{210}$ \\\hline
   2       &    $\mathcal{L}_{17}$  &    $\mathcal{L}_{21}$    & $ \mathcal{L}_{25}$   &  $\mathcal{L}_{29}$   & 2   &    $\mathcal{L}_{77}$  &    $\mathcal{L}_{81}$   & $ \mathcal{L}_{85}$   &  $\mathcal{L}_{89}$              & 2   &    $\mathcal{L}_{137}$    &    $\mathcal{L}_{141}$    & $ \mathcal{L}_{145}$    &  $\mathcal{L}_{149}$   &  2    &    $\mathcal{L}_{199}$             &    $\mathcal{L}_{203}$    & $ \mathcal{L}_{207}$   &  $\mathcal{L}_{211}$ \\\hline
   3       &    $\mathcal{L}_{18}$  &    $\mathcal{L}_{22}$    & $ \mathcal{L}_{26}$   &  $\mathcal{L}_{30}$   & 3   &    $\mathcal{L}_{78}$  &    $\mathcal{L}_{82}$   & $\mathcal{L}_{86}$    &  $\mathcal{L}_{90}$              & 3   &    $\mathcal{L}_{138}$    &    $\mathcal{L}_{142}$    & $\mathcal{L}_{146}$     &  $\mathcal{L}_{150}$   &  3    &    $\mathcal{L}_{200}$             &    $\mathcal{L}_{204}$    & $ \mathcal{L}_{208}$   &  $\mathcal{L}_{212}$ \\\noalign{\hrule height 1.pt}
 $x_C=1$   &\multirow{2}{*}{0} & \multirow{2}{*}{1} & \multirow{2}{*}{2} & \multirow{2}{*}{3}                                   & $x_C=1$ & \multirow{2}{*}{0} & \multirow{2}{*}{1} & \multirow{2}{*}{2} & \multirow{2}{*}{3}                & $x_C=1$ & \multirow{2}{*}{0} & \multirow{2}{*}{1} & \multirow{2}{*}{2} & \multirow{2}{*}{3}                   &  $x_C=1$ & \multirow{2}{*}{0} & \multirow{2}{*}{1} & \multirow{2}{*}{2} & \multirow{2}{*}{3} \\
 $x_D=0$   &   &   &   &                                                                                                        & $x_D=1$ &   &   &   &                                                                                      & $x_D=2$ &   &   &   &                                                                                         &  $x_D=3$ &   &   &   &   \\\noalign{\hrule height 0.9pt}
   0   &    $\mathcal{L}_{31}$    &    $\mathcal{L}_{16}$    & $ \mathcal{L}_{37}$   &  $\mathcal{L}_{41}$    & 0   &    $\mathcal{L}_{91}$  &    $\mathcal{L}_{14}$    & $ \mathcal{L}_{97}$   &  $\mathcal{L}_{101}$    &0   &    $\mathcal{L}_{151}$       &    $\mathcal{L}_{155}$    & $ \mathcal{L}_{159}$   &  $\mathcal{L}_{163}$           & 0   &    $\mathcal{L}_{213}$            &    $\mathcal{L}_{217}$              & $ \mathcal{L}_{221}$   &  $\mathcal{L}_{225}$ \\\hline
   1   &    $\mathcal{L}_{32}$    &    $\mathcal{L}_{35}$    & $ \mathcal{L}_{38}$   &  $\mathcal{L}_{42}$    & 1   &    $\mathcal{L}_{92}$  &    $\mathcal{L}_{95}$    & $ \mathcal{L}_{98}$   &  $\mathcal{L}_{102}$    &1   &    $\mathcal{L}_{152}$       &    $\mathcal{L}_{156}$    & $ \mathcal{L}_{160}$   &  $\mathcal{L}_{164}$           & 1   &    $\mathcal{L}_{214}$            &    $\mathcal{L}_{218}$              & $ \mathcal{L}_{222}$   &  $\mathcal{L}_{226}$ \\\hline
   2   &    $\mathcal{L}_{33}$    &    $\mathcal{L}_{12}$    & $ \mathcal{L}_{39}$   &  $\mathcal{L}_{43}$    & 2   &    $\mathcal{L}_{93}$  &    $\mathcal{L}_{10}$    & $ \mathcal{L}_{99}$   &  $\mathcal{L}_{103}$    &2   &    $\mathcal{L}_{153}$       &    $\mathcal{L}_{157}$    & $ \mathcal{L}_{161}$   &  $\mathcal{L}_{165}$           & 2   &    $\mathcal{L}_{215}$            &    $\mathcal{L}_{219}$              & $ \mathcal{L}_{223}$   &  $\mathcal{L}_{227}$ \\\hline
   3   &    $\mathcal{L}_{34}$    &    $\mathcal{L}_{36}$    & $ \mathcal{L}_{40}$   &  $\mathcal{L}_{44}$    & 3   &    $\mathcal{L}_{94}$  &    $\mathcal{L}_{96}$    & $ \mathcal{L}_{100}$  &  $\mathcal{L}_{104}$    &3   &    $\mathcal{L}_{154}$       &    $\mathcal{L}_{158}$    & $ \mathcal{L}_{162}$   &  $\mathcal{L}_{166}$           & 3   &    $\mathcal{L}_{216}$              &    $\mathcal{L}_{220}$            & $ \mathcal{L}_{224}$   &  $\mathcal{L}_{228}$ \\\noalign{\hrule height 1.pt}
 $x_C=2$   &\multirow{2}{*}{0} & \multirow{2}{*}{1} & \multirow{2}{*}{2} & \multirow{2}{*}{3}                                            & $x_C=2$ & \multirow{2}{*}{0} & \multirow{2}{*}{1} & \multirow{2}{*}{2} & \multirow{2}{*}{3}            & $x_C=2$ & \multirow{2}{*}{0} & \multirow{2}{*}{1} & \multirow{2}{*}{2} & \multirow{2}{*}{3}                                    &  $x_C=2$ & \multirow{2}{*}{0} & \multirow{2}{*}{1} & \multirow{2}{*}{2} & \multirow{2}{*}{3} \\
 $x_D=0$   &   &   &   &                                                                                                                 & $x_D=1$ &   &   &   &                                                                                  & $x_D=2$ &   &   &   &                                                                                                          &  $x_D=3$ &   &   &   &   \\\noalign{\hrule height 0.9pt}
 0   &    $\mathcal{L}_{45}$  &    $\mathcal{L}_{49}$    & $ \mathcal{L}_{53}$  &  $\mathcal{L}_{57}$         &0   &    $\mathcal{L}_{105}$  &    $\mathcal{L}_{109}$    & $ \mathcal{L}_{113}$  &  $\mathcal{L}_{117}$  & 0   &    $\mathcal{L}_{167}$  &    $\mathcal{L}_{171}$    & $ \mathcal{L}_{175}$   &  $\mathcal{L}_{177}$     & 0   &    $\mathcal{L}_{229}$              &    $\mathcal{L}_{233}$               & $ \mathcal{L}_{237}$   &  $\mathcal{L}_{239}$ \\\hline
 1   &    $\mathcal{L}_{46}$  &    $\mathcal{L}_{50}$    & $ \mathcal{L}_{54}$  &  $\mathcal{L}_{58}$         &1   &    $\mathcal{L}_{106}$  &    $\mathcal{L}_{110}$    & $ \mathcal{L}_{114}$  &  $\mathcal{L}_{118}$  & 1   &    $\mathcal{L}_{168}$  &    $\mathcal{L}_{172}$    & $ \mathcal{L}_{11}$   &  $\mathcal{L}_{178}$     & 1   &    $\mathcal{L}_{230}$              &    $\mathcal{L}_{234}$                & $ \mathcal{L}_{9}$   &  $\mathcal{L}_{240}$ \\\hline
 2   &    $\mathcal{L}_{47}$  &    $\mathcal{L}_{51}$    & $\mathcal{L}_{55}$   &  $\mathcal{L}_{59}$         &2   &    $\mathcal{L}_{107}$  &    $\mathcal{L}_{111}$    & $\mathcal{L}_{115}$   &  $\mathcal{L}_{119}$  & 2   &    $\mathcal{L}_{169}$  &    $\mathcal{L}_{172}$    & $ \mathcal{L}_{15}$   &  $\mathcal{L}_{179}$     & 2   &    $\mathcal{L}_{231}$              &    $\mathcal{L}_{235}$                 & $ \mathcal{L}_{13}$   &  $\mathcal{L}_{241}$ \\\hline
 3   &    $\mathcal{L}_{48}$  &    $\mathcal{L}_{52}$    & $\mathcal{L}_{56}$   &  $\mathcal{L}_{60}$         &3   &    $\mathcal{L}_{108}$  &    $\mathcal{L}_{112}$    & $\mathcal{L}_{116}$   &  $\mathcal{L}_{120}$  & 3   &    $\mathcal{L}_{170}$  &    $\mathcal{L}_{174}$    & $ \mathcal{L}_{176}$   &  $\mathcal{L}_{180}$     & 3   &    $\mathcal{L}_{232}$              &    $\mathcal{L}_{236}$                & $ \mathcal{L}_{238}$   &  $\mathcal{L}_{242}$ \\\noalign{\hrule height 1.pt}

 $x_C=3$   &\multirow{2}{*}{0} & \multirow{2}{*}{1} & \multirow{2}{*}{2} & \multirow{2}{*}{3}          & $x_C=3$ & \multirow{2}{*}{0} & \multirow{2}{*}{1} & \multirow{2}{*}{2} & \multirow{2}{*}{3}                                                & $x_C=3$ & \multirow{2}{*}{0} & \multirow{2}{*}{1} & \multirow{2}{*}{2} & \multirow{2}{*}{3}                                               &  $x_C=3$ & \multirow{2}{*}{0} & \multirow{2}{*}{1} & \multirow{2}{*}{2} & \multirow{2}{*}{3} \\
 $x_D=0$   &   &   &   &                                                                               & $x_D=1$ &   &   &   &                                                                                                                      & $x_D=2$ &   &   &   &                                                                                                                     &  $x_D=3$ &   &   &   &   \\\noalign{\hrule height 0.9pt}
 0   &    $\mathcal{L}_{61}$  &    $\mathcal{L}_{65}$    & $ \mathcal{L}_{69}$  &  $\mathcal{L}_{73}$ & 0   &   $\mathcal{L}_{121}$  &    $\mathcal{L}_{125}$   & $\mathcal{L}_{129}$    &  $\mathcal{L}_{133}$     & 0   &    $\mathcal{L}_{181}$  &    $\mathcal{L}_{185}$    & $\mathcal{L}_{189}$      &  $\mathcal{L}_{193}$          &0   &    $\mathcal{L}_{243}$  &    $\mathcal{L}_{247}$    & $ \mathcal{L}_{251}$   &  $\mathcal{L}_{255}$ \\\hline 
 1   &    $\mathcal{L}_{62}$  &    $\mathcal{L}_{66}$    & $ \mathcal{L}_{70}$  &  $\mathcal{L}_{74}$ & 1   &   $\mathcal{L}_{122}$  &    $\mathcal{L}_{126}$   & $ \mathcal{L}_{130}$   &  $\mathcal{L}_{134}$     & 1   &    $\mathcal{L}_{182}$  &    $\mathcal{L}_{186}$    & $ \mathcal{L}_{190}$     &  $\mathcal{L}_{194}$          &1   &    $\mathcal{L}_{244}$  &    $\mathcal{L}_{248}$    & $ \mathcal{L}_{252}$   &  $\mathcal{L}_{256}$ \\\hline 
 2   &    $\mathcal{L}_{63}$  &    $\mathcal{L}_{67}$    & $\mathcal{L}_{71}$   &  $\pmb{\mathcal{L}}_{\pmb{1}}$ & 2   &    $\mathcal{L}_{123}$  &    $\mathcal{L}_{127}$    & $\mathcal{L}_{131}$    &  $\mathcal{L}_{135}$    & 2   &    $\mathcal{L}_{183}$  &    $\mathcal{L}_{187}$    & $ \mathcal{L}_{191}$     &  $\mathcal{L}_{195}$          &2   &    $\mathcal{L}_{245}$  &    $\mathcal{L}_{249}$    & $ \mathcal{L}_{253}$   &  $\pmb{\mathcal{L}}_{\pmb{3}}$ \\\hline 
 3   &    $\mathcal{L}_{64}$  &    $\mathcal{L}_{68}$    & $\mathcal{L}_{72}$   &  $\pmb{\mathcal{L}}_{\pmb{5}}$ & 3   &    $\mathcal{L}_{124}$  &    $\mathcal{L}_{128}$    & $\mathcal{L}_{132}$    &  $\mathcal{L}_{136}$    & 3   &    $\mathcal{L}_{184}$  &    $\mathcal{L}_{188}$    & $\mathcal{L}_{192}$      &  $\mathcal{L}_{196}$          &3   &    $\mathcal{L}_{246}$  &    $\mathcal{L}_{250}$    & $ \mathcal{L}_{254}$   &  $\pmb{\mathcal{L}}_{\pmb{7}}$ \\\noalign{\hrule height 1.pt}
\end{tabular}
\caption{Example 1: Latin-Hyper Cube representing the Relay Map where $x_E=0$, entries $x_C$'s and $x_D$'s are as mentioned, $x_A$'s entries are along the rows and $x_B$'s entries are along the columns of each $4\times 4$ matrix for fixed values of $x_C$ and $x_D$.}
\vspace{-.4cm}
\end{table*}

\begin{table*}[ht]
\vspace{-. cm}
\tiny
\centering
\renewcommand{\arraystretch}{1.2}
\begin{tabular}{!{\vrule width 1.pt}c!{\vrule width 0.9pt}p{0.32cm}|p{0.32cm}|p{0.32cm}|p{0.32cm}!{\vrule width 1.pt}c!{\vrule width 0.9pt}p{0.3cm}|p{0.32cm}|p{0.32cm}|p{0.32cm}!{\vrule width 1.pt}c!{\vrule width 0.9pt}p{0.32cm}|p{0.32cm}|p{0.32cm}|p{0.32cm}!{\vrule width 1.pt}c!{\vrule width 0.9pt}p{0.32cm}|p{0.32cm}|p{0.32cm}|p{0.32cm}!{\vrule width 1.pt}}\noalign{\hrule height 1.pt}
 $x_C=0$   &\multirow{2}{*}{0} & \multirow{2}{*}{1} & \multirow{2}{*}{2} & \multirow{2}{*}{3}                  & $x_C=0$ & \multirow{2}{*}{0} & \multirow{2}{*}{1} & \multirow{2}{*}{2} & \multirow{2}{*}{3}                                    & $x_C=0$ & \multirow{2}{*}{0} & \multirow{2}{*}{1} & \multirow{2}{*}{2} & \multirow{2}{*}{3}            &  $x_C=0$ & \multirow{2}{*}{0} & \multirow{2}{*}{1} & \multirow{2}{*}{2} & \multirow{2}{*}{3} \\
 $x_D=0$   &   &   &   &                                                                                       & $x_D=1$ &   &   &   &                                                                                                          & $x_D=2$ &   &   &   &                                                                                  &  $x_D=3$ &   &   &   &   \\\noalign{\hrule height 0.9pt}
   0       &    $\mathcal{L}_{9}$  &    $\mathcal{L}_{92}$    & $ \mathcal{L}_{3}$   &  $\mathcal{L}_{98}$   & 0   &    $\mathcal{L}_{11}$  &    $\mathcal{L}_{32}$    & $\mathcal{L}_{42}$   &  $\mathcal{L}_{38}$               & 0   &    $\mathcal{L}_{1}$    &    $\mathcal{L}_{135}$    & $\mathcal{L}_{215}$     &  $\mathcal{L}_{219}$  &  0       &    $\mathcal{L}_{156}$             &    $\mathcal{L}_{152}$    & $ \mathcal{L}_{164}$   &  $\mathcal{L}_{160}$ \\\hline
   1       &    $\mathcal{L}_{7}$  &    $\mathcal{L}_{91}$    & $ \mathcal{L}_{96}$   &  $\mathcal{L}_{97}$   & 1   &    $\mathcal{L}_{36}$  &    $\mathcal{L}_{31}$    & $ \mathcal{L}_{41}$    &  $\mathcal{L}_{37}$            & 1   &    $\mathcal{L}_{5}$    &    $\mathcal{L}_{136}$    & $ \mathcal{L}_{216}$    &  $\mathcal{L}_{220}$  &  1       &    $\mathcal{L}_{155}$             &    $\mathcal{L}_{151}$    & $ \mathcal{L}_{163}$   &  $\mathcal{L}_{159}$ \\\hline
   2       &    $\mathcal{L}_{95}$  &    $\mathcal{L}_{94}$    & $ \mathcal{L}_{101}$   &  $\mathcal{L}_{100}$   & 2   &    $\mathcal{L}_{35}$  &    $\mathcal{L}_{34}$   & $ \mathcal{L}_{44}$   &  $\mathcal{L}_{40}$           & 2   &    $\mathcal{L}_{116}$    &    $\mathcal{L}_{213}$    & $ \mathcal{L}_{217}$    &  $\mathcal{L}_{221}$   &  2    &    $\mathcal{L}_{158}$             &    $\mathcal{L}_{154}$    & $ \mathcal{L}_{166}$   &  $\mathcal{L}_{162}$ \\\hline
   3       &    $\mathcal{L}_{13}$  &    $\mathcal{L}_{93}$    & $ \mathcal{L}_{102}$   &  $\mathcal{L}_{99}$   & 3   &    $\mathcal{L}_{15}$  &    $\mathcal{L}_{33}$   & $\mathcal{L}_{43}$    &  $\mathcal{L}_{39}$            & 3   &    $\mathcal{L}_{115}$    &    $\mathcal{L}_{214}$    & $\mathcal{L}_{218}$     &  $\mathcal{L}_{222}$   &  3    &    $\mathcal{L}_{157}$             &    $\mathcal{L}_{153}$    & $ \mathcal{L}_{165}$   &  $\mathcal{L}_{161}$ \\\noalign{\hrule height 1.pt}
 $x_C=1$   &\multirow{2}{*}{0} & \multirow{2}{*}{1} & \multirow{2}{*}{2} & \multirow{2}{*}{3}                                   & $x_C=1$ & \multirow{2}{*}{0} & \multirow{2}{*}{1} & \multirow{2}{*}{2} & \multirow{2}{*}{3}                & $x_C=1$ & \multirow{2}{*}{0} & \multirow{2}{*}{1} & \multirow{2}{*}{2} & \multirow{2}{*}{3}                   &  $x_C=1$ & \multirow{2}{*}{0} & \multirow{2}{*}{1} & \multirow{2}{*}{2} & \multirow{2}{*}{3} \\
 $x_D=0$   &   &   &   &                                                                                                        & $x_D=1$ &   &   &   &                                                                                      & $x_D=2$ &   &   &   &                                                                                         &  $x_D=3$ &   &   &   &   \\\noalign{\hrule height 0.9pt}
   0   &    $\mathcal{L}_{80}$    &    $\mathcal{L}_{76}$    & $ \mathcal{L}_{88}$   &  $\mathcal{L}_{84}$    & 0   &    $\mathcal{L}_{20}$  &    $\mathcal{L}_{17}$    & $ \mathcal{L}_{28}$   &  $\mathcal{L}_{25}$    &0   &    $\mathcal{L}_{202}$       &    $\mathcal{L}_{198}$    & $ \mathcal{L}_{210}$   &  $\mathcal{L}_{206}$           & 0   &    $\mathcal{L}_{141}$            &    $\mathcal{L}_{145}$              & $ \mathcal{L}_{149}$   &  $\mathcal{L}_{169}$ \\\hline
   1   &    $\mathcal{L}_{79}$    &    $\mathcal{L}_{75}$    & $ \mathcal{L}_{87}$   &  $\mathcal{L}_{83}$    & 1   &    $\mathcal{L}_{19}$  &    $\mathcal{L}_{18}$    & $ \mathcal{L}_{27}$   &  $\mathcal{L}_{26}$    &1   &    $\mathcal{L}_{201}$       &    $\mathcal{L}_{197}$    & $ \mathcal{L}_{209}$   &  $\mathcal{L}_{205}$           & 1   &    $\mathcal{L}_{142}$            &    $\mathcal{L}_{146}$              & $ \mathcal{L}_{150}$   &  $\mathcal{L}_{170}$ \\\hline
   2   &    $\mathcal{L}_{82}$    &    $\mathcal{L}_{78}$    & $ \mathcal{L}_{90}$   &  $\mathcal{L}_{86}$    & 2   &    $\mathcal{L}_{22}$  &    $\mathcal{L}_{23}$    & $ \mathcal{L}_{30}$   &  $\mathcal{L}_{45}$    &2   &    $\mathcal{L}_{204}$       &    $\mathcal{L}_{200}$    & $ \mathcal{L}_{212}$   &  $\mathcal{L}_{208}$           & 2   &    $\mathcal{L}_{143}$            &    $\mathcal{L}_{147}$              & $ \mathcal{L}_{167}$   &  $\mathcal{L}_{171}$ \\\hline
   3   &    $\mathcal{L}_{81}$    &    $\mathcal{L}_{77}$    & $ \mathcal{L}_{89}$   &  $\mathcal{L}_{85}$    & 3   &    $\mathcal{L}_{21}$  &    $\mathcal{L}_{24}$    & $ \mathcal{L}_{29}$  &  $\mathcal{L}_{46}$    &3   &    $\mathcal{L}_{203}$       &    $\mathcal{L}_{199}$    & $ \mathcal{L}_{211}$   &  $\mathcal{L}_{207}$           & 3   &    $\mathcal{L}_{144}$              &    $\mathcal{L}_{148}$            & $ \mathcal{L}_{168}$   &  $\mathcal{L}_{172}$ \\\noalign{\hrule height 1.pt}
 $x_C=2$   &\multirow{2}{*}{0} & \multirow{2}{*}{1} & \multirow{2}{*}{2} & \multirow{2}{*}{3}                                            & $x_C=2$ & \multirow{2}{*}{0} & \multirow{2}{*}{1} & \multirow{2}{*}{2} & \multirow{2}{*}{3}            & $x_C=2$ & \multirow{2}{*}{0} & \multirow{2}{*}{1} & \multirow{2}{*}{2} & \multirow{2}{*}{3}                                    &  $x_C=2$ & \multirow{2}{*}{0} & \multirow{2}{*}{1} & \multirow{2}{*}{2} & \multirow{2}{*}{3} \\
 $x_D=0$   &   &   &   &                                                                                                                 & $x_D=1$ &   &   &   &                                                                                  & $x_D=2$ &   &   &   &                                                                                                          &  $x_D=3$ &   &   &   &   \\\noalign{\hrule height 0.9pt}
 0   &    $\mathcal{L}_{103}$  &    $\mathcal{L}_{122}$    & $ \mathcal{L}_{127}$  &  $\mathcal{L}_{130}$         &0   &    $\mathcal{L}_{66}$  &    $\mathcal{L}_{62}$    & $ \mathcal{L}_{74}$  &  $\mathcal{L}_{70}$  & 0   &    $\mathcal{L}_{223}$  &    $\mathcal{L}_{227}$    & $ \mathcal{L}_{245}$   &  $\mathcal{L}_{249}$       & 0   &    $\mathcal{L}_{186}$              &    $\mathcal{L}_{182}$               & $ \mathcal{L}_{194}$   &  $\mathcal{L}_{190}$ \\\hline
 1   &    $\mathcal{L}_{104}$  &    $\mathcal{L}_{121}$    & $ \mathcal{L}_{128}$  &  $\mathcal{L}_{129}$         &1   &    $\mathcal{L}_{65}$  &    $\mathcal{L}_{61}$    & $ \mathcal{L}_{73}$  &  $\mathcal{L}_{69}$  & 1   &    $\mathcal{L}_{224}$  &    $\mathcal{L}_{228}$    & $ \mathcal{L}_{246}$   &  $\mathcal{L}_{250}$       & 1   &    $\mathcal{L}_{185}$              &    $\mathcal{L}_{181}$                & $ \mathcal{L}_{193}$   &  $\mathcal{L}_{189}$ \\\hline
 2   &    $\mathcal{L}_{125}$  &    $\mathcal{L}_{124}$    & $\mathcal{L}_{133}$   &  $\mathcal{L}_{132}$         &2   &    $\mathcal{L}_{68}$  &    $\mathcal{L}_{64}$    & $\mathcal{L}_{138}$   &  $\mathcal{L}_{72}$  & 2   &    $\mathcal{L}_{225}$  &    $\mathcal{L}_{243}$    & $ \mathcal{L}_{247}$   &  $\mathcal{L}_{251}$      & 2   &    $\mathcal{L}_{188}$              &    $\mathcal{L}_{184}$                 & $ \mathcal{L}_{196}$   &  $\mathcal{L}_{192}$ \\\hline
 3   &    $\mathcal{L}_{126}$  &    $\mathcal{L}_{123}$    & $\mathcal{L}_{134}$   &  $\mathcal{L}_{131}$         &3   &    $\mathcal{L}_{67}$  &    $\mathcal{L}_{63}$    & $\mathcal{L}_{137}$   &  $\mathcal{L}_{71}$  & 3   &    $\mathcal{L}_{226}$  &    $\mathcal{L}_{244}$    & $ \mathcal{L}_{248}$   &  $\mathcal{L}_{252}$      & 3   &    $\mathcal{L}_{187}$              &    $\mathcal{L}_{183}$                & $ \mathcal{L}_{195}$   &  $\mathcal{L}_{191}$ \\\noalign{\hrule height 1.pt}

 $x_C=3$   &\multirow{2}{*}{0} & \multirow{2}{*}{1} & \multirow{2}{*}{2} & \multirow{2}{*}{3}          & $x_C=3$ & \multirow{2}{*}{0} & \multirow{2}{*}{1} & \multirow{2}{*}{2} & \multirow{2}{*}{3}                                                & $x_C=3$ & \multirow{2}{*}{0} & \multirow{2}{*}{1} & \multirow{2}{*}{2} & \multirow{2}{*}{3}                                               &  $x_C=3$ & \multirow{2}{*}{0} & \multirow{2}{*}{1} & \multirow{2}{*}{2} & \multirow{2}{*}{3} \\
 $x_D=0$   &   &   &   &                                                                               & $x_D=1$ &   &   &   &                                                                                                                      & $x_D=2$ &   &   &   &                                                                                                                     &  $x_D=3$ &   &   &   &   \\\noalign{\hrule height 0.9pt}
 0   &    $\mathcal{L}_{110}$  &    $\mathcal{L}_{106}$    & $ \mathcal{L}_{118}$  &  $\mathcal{L}_{114}$ & 0   &   $\mathcal{L}_{50}$  &    $\mathcal{L}_{47}$   & $\mathcal{L}_{58}$    &  $\mathcal{L}_{55}$     & 0   &    $\mathcal{L}_{234}$  &    $\mathcal{L}_{230}$    & $\mathcal{L}_{240}$      &  $\mathcal{L}_{238}$          &0   &    $\mathcal{L}_{173}$  &    $\mathcal{L}_{176}$    & $ \mathcal{L}_{180}$   &  $\mathcal{L}_{261}$ \\\hline 
 1   &    $\mathcal{L}_{109}$  &    $\mathcal{L}_{105}$    & $ \mathcal{L}_{117}$  &  $\mathcal{L}_{113}$ & 1   &   $\mathcal{L}_{49}$  &    $\mathcal{L}_{48}$   & $ \mathcal{L}_{57}$   &  $\mathcal{L}_{56}$     & 1   &    $\mathcal{L}_{233}$  &    $\mathcal{L}_{229}$    & $ \mathcal{L}_{239}$     &  $\mathcal{L}_{14}$          &1   &    $\mathcal{L}_{174}$  &    $\mathcal{L}_{179}$    & $ \mathcal{L}_{258}$   &  $\mathcal{L}_{16}$ \\\hline 
 2   &    $\mathcal{L}_{112}$  &    $\mathcal{L}_{108}$    & $\mathcal{L}_{120}$   &  $\pmb{\mathcal{L}}_{\pmb{2}}$ & 2   &    $\mathcal{L}_{52}$  &    $\mathcal{L}_{53}$    & $\mathcal{L}_{60}$    &  $\mathcal{L}_{139}$    & 2   &    $\mathcal{L}_{236}$  &    $\mathcal{L}_{232}$    & $ \mathcal{L}_{242}$     &  $\mathcal{L}_{237}$          &2   &    $\mathcal{L}_{175}$  &    $\mathcal{L}_{178}$    & $ \mathcal{L}_{259}$   &  $\pmb{\mathcal{L}}_{\pmb{4}}$ \\\hline 
 3   &    $\mathcal{L}_{111}$  &    $\mathcal{L}_{107}$    & $\mathcal{L}_{119}$   & $\pmb{\mathcal{L}}_{\pmb{6}}$ & 3   &    $\mathcal{L}_{51}$  &    $\mathcal{L}_{54}$    & $\mathcal{L}_{59}$    &  $\mathcal{L}_{140}$    & 3   &    $\mathcal{L}_{235}$  &    $\mathcal{L}_{231}$    & $\mathcal{L}_{241}$      &  $\mathcal{L}_{10}$          &3   &    $\mathcal{L}_{177}$  &    $\mathcal{L}_{257}$    & $ \mathcal{L}_{260}$   &  $\pmb{\mathcal{L}}_{\pmb{8}}$ \\\noalign{\hrule height 1.pt}
\end{tabular}
\caption{Example 1: Latin-Hyper Cube representing the Relay Map where $x_E=1$, entries $x_A$'s and $x_B$'s are as mentioned, $x_C$'s entries are along the rows and $x_D$'s entries are along the columns of each $4\times 4$ matrix for fixed values of $x_A$ and $x_B$.}
\vspace{-.9cm}
\end{table*}
\begin{table*}[ht]
\vspace{-. cm}
\tiny
\centering
\renewcommand{\arraystretch}{1.2}
\begin{tabular}{!{\vrule width 1.pt}c!{\vrule width 0.9pt}p{0.32cm}|p{0.32cm}|p{0.32cm}|p{0.3cm}!{\vrule width 1.pt}c!{\vrule width 0.9pt}p{0.3cm}|p{0.32cm}|p{0.32cm}|p{0.32cm}!{\vrule width 1.pt}c!{\vrule width 0.9pt}p{0.3cm}|p{0.3cm}|p{0.3cm}|p{0.32cm}!{\vrule width 1.pt}c!{\vrule width 0.9pt}p{0.32cm}|p{0.32cm}|p{0.32cm}|p{0.32cm}!{\vrule width 1.pt}}\noalign{\hrule height 1.pt}
 $x_C=0$   &\multirow{2}{*}{0} & \multirow{2}{*}{1} & \multirow{2}{*}{2} & \multirow{2}{*}{3}                  & $x_C=0$ & \multirow{2}{*}{0} & \multirow{2}{*}{1} & \multirow{2}{*}{2} & \multirow{2}{*}{3}                                    & $x_C=0$ & \multirow{2}{*}{0} & \multirow{2}{*}{1} & \multirow{2}{*}{2} & \multirow{2}{*}{3}            &  $x_C=0$ & \multirow{2}{*}{0} & \multirow{2}{*}{1} & \multirow{2}{*}{2} & \multirow{2}{*}{3} \\
 $x_D=0$   &   &   &   &                                                                                       & $x_D=1$ &   &   &   &                                                                                                          & $x_D=2$ &   &   &   &                                                                                  &  $x_D=3$ &   &   &   &   \\\noalign{\hrule height 0.9pt}
   0       &    $\mathcal{L}_{10}$  &    $\mathcal{L}_{168}$    & $ \mathcal{L}_{174}$   &  $\mathcal{L}_{183}$   & 0   &    $\mathcal{L}_{12}$  &    $\mathcal{L}_{191}$    & $\mathcal{L}_{187}$   &  $\mathcal{L}_{224}$               & 0   &    $\mathcal{L}_{54}$    &    $\mathcal{L}_{46}$    & $\mathcal{L}_{48}$     &  $\mathcal{L}_{51}$  &  0    &    $\mathcal{L}_{71}$           &    $\mathcal{L}_{104}$    & $ \mathcal{L}_{107}$   &  $\mathcal{L}_{111}$ \\\hline
   1       &    $\mathcal{L}_{171}$  &    $\mathcal{L}_{167}$    & $ \mathcal{L}_{173}$   &  $\mathcal{L}_{175}$   & 1   &    $\mathcal{L}_{180}$  &    $\mathcal{L}_{192}$    & $ \mathcal{L}_{188}$    &  $\mathcal{L}_{223}$              & 1   &   $\mathcal{L}_{53}$    &    $\mathcal{L}_{45}$    & $ \mathcal{L}_{47}$    &  $\mathcal{L}_{52}$  &  1  &    $\mathcal{L}_{72}$           &    $\mathcal{L}_{103}$    & $ \mathcal{L}_{108}$   &  $\mathcal{L}_{112}$ \\\hline
   2       &    $\mathcal{L}_{172}$  &    $\mathcal{L}_{170}$    & $ \mathcal{L}_{177}$   &  $\mathcal{L}_{176}$   & 2   &    $\mathcal{L}_{190}$  &    $\mathcal{L}_{193}$   & $ \mathcal{L}_{226}$   &  $\mathcal{L}_{230}$              & 2   &    $\mathcal{L}_{56}$    &    $\mathcal{L}_{57}$    & $ \mathcal{L}_{49}$    &  $\mathcal{L}_{61}$   &  2  &    $\mathcal{L}_{113}$          &    $\mathcal{L}_{117}$    & $ \mathcal{L}_{109}$   &  $\mathcal{L}_{121}$ \\\hline
   3       &    $\mathcal{L}_{14}$  &    $\mathcal{L}_{169}$    & $ \mathcal{L}_{178}$   &  $\mathcal{L}_{181}$   & 3   &    $\mathcal{L}_{16}$  &    $\mathcal{L}_{194}$   & $\mathcal{L}_{225}$    &  $\mathcal{L}_{233}$              & 3   &    $\mathcal{L}_{55}$    &    $\mathcal{L}_{58}$    & $\mathcal{L}_{50}$     &  $\mathcal{L}_{62}$   &  3    &    $\mathcal{L}_{114}$          &    $\mathcal{L}_{118}$    & $ \mathcal{L}_{110}$   &  $\mathcal{L}_{122}$ \\\noalign{\hrule height 1.pt}
 $x_C=1$   &\multirow{2}{*}{0} & \multirow{2}{*}{1} & \multirow{2}{*}{2} & \multirow{2}{*}{3}                                   & $x_C=1$ & \multirow{2}{*}{0} & \multirow{2}{*}{1} & \multirow{2}{*}{2} & \multirow{2}{*}{3}                & $x_C=1$ & \multirow{2}{*}{0} & \multirow{2}{*}{1} & \multirow{2}{*}{2} & \multirow{2}{*}{3}                   &  $x_C=1$ & \multirow{2}{*}{0} & \multirow{2}{*}{1} & \multirow{2}{*}{2} & \multirow{2}{*}{3} \\
 $x_D=0$   &   &   &   &                                                                                                        & $x_D=1$ &   &   &   &                                                                                      & $x_D=2$ &   &   &   &                                                                                         &  $x_D=3$ &   &   &   &   \\\noalign{\hrule height 0.9pt}
   0   &    $\mathcal{L}_{136}$    &    $\mathcal{L}_{15}$    & $ \mathcal{L}_{175}$   &  $\mathcal{L}_{116}$    & 0   &    $\mathcal{L}_{195}$  &    $\pmb{\mathcal{L}}_{\pmb{6}}$    & $ \mathcal{L}_{231}$   &  $\mathcal{L}_{13}$    &0   &    $\mathcal{L}_{59}$       &    $\pmb{\mathcal{L}}_{\pmb{8}}$    & $ \mathcal{L}_{63}$   &  $\mathcal{L}_{67}$           & 0   &    $\mathcal{L}_{119}$            &    $\mathcal{L}_{131}$              & $ \mathcal{L}_{123}$   &  $\mathcal{L}_{269}$ \\\hline
   1   &    $\mathcal{L}_{135}$    &    $\mathcal{L}_{115}$    & $ \mathcal{L}_{184}$   &  $\mathcal{L}_{137}$    & 1   &    $\mathcal{L}_{196}$  &   $\pmb{\mathcal{L}}_{\pmb{2}}$    & $ \mathcal{L}_{232}$   &  $\mathcal{L}_{235}$    &1   &    $\mathcal{L}_{60}$       &   $\pmb{\mathcal{L}}_{\pmb{4}}$    & $ \mathcal{L}_{64}$   &  $\mathcal{L}_{68}$           & 1   &    $\mathcal{L}_{120}$            &    $\mathcal{L}_{132}$              & $ \mathcal{L}_{124}$   &  $\mathcal{L}_{270}$ \\\hline
   2   &    $\mathcal{L}_{140}$    &    $\mathcal{L}_{11}$    & $ \mathcal{L}_{182}$   &  $\mathcal{L}_{186}$    & 2   &    $\mathcal{L}_{238}$  &    $\mathcal{L}_{9}$    & $ \mathcal{L}_{234}$   &  $\mathcal{L}_{244}$    &2   &    $\mathcal{L}_{69}$       &    $\mathcal{L}_{73}$    & $ \mathcal{L}_{65}$   &  $\mathcal{L}_{105}$           & 2   &    $\mathcal{L}_{128}$            &    $\mathcal{L}_{129}$              & $ \mathcal{L}_{126}$   &  $\mathcal{L}_{271}$ \\\hline
   3   &    $\mathcal{L}_{139}$    &    $\mathcal{L}_{189}$    & $ \mathcal{L}_{185}$   &  $\mathcal{L}_{229}$    & 3   &    $\mathcal{L}_{237}$  &    $\mathcal{L}_{239}$    & $ \mathcal{L}_{243}$  &  $\mathcal{L}_{245}$    &3   &    $\mathcal{L}_{70}$       &    $\mathcal{L}_{74}$    & $ \mathcal{L}_{66}$   &  $\mathcal{L}_{106}$           & 3   &    $\mathcal{L}_{127}$              &    $\mathcal{L}_{130}$            & $ \mathcal{L}_{125}$   &  $\mathcal{L}_{272}$ \\\noalign{\hrule height 1.pt}
 $x_C=2$   &\multirow{2}{*}{0} & \multirow{2}{*}{1} & \multirow{2}{*}{2} & \multirow{2}{*}{3}                                            & $x_C=2$ & \multirow{2}{*}{0} & \multirow{2}{*}{1} & \multirow{2}{*}{2} & \multirow{2}{*}{3}            & $x_C=2$ & \multirow{2}{*}{0} & \multirow{2}{*}{1} & \multirow{2}{*}{2} & \multirow{2}{*}{3}                                    &  $x_C=2$ & \multirow{2}{*}{0} & \multirow{2}{*}{1} & \multirow{2}{*}{2} & \multirow{2}{*}{3} \\
 $x_D=0$   &   &   &   &                                                                                                                 & $x_D=1$ &   &   &   &                                                                                  & $x_D=2$ &   &   &   &                                                                                                          &  $x_D=3$ &   &   &   &   \\\noalign{\hrule height 0.9pt}
 0   &    $\mathcal{L}_{146}$  &    $\mathcal{L}_{144}$    & $ \mathcal{L}_{142}$  &  $\mathcal{L}_{153}$         &0   &    $\mathcal{L}_{166}$  &    $\mathcal{L}_{208}$    & $ \mathcal{L}_{203}$  &  $\mathcal{L}_{214}$  & 0   &    $\mathcal{L}_{24}$  &    $\mathcal{L}_{29}$    & $ \mathcal{L}_{18}$   &  $\mathcal{L}_{21}$     & 0   &    $\mathcal{L}_{85}$              &    $\mathcal{L}_{89}$               & $ \mathcal{L}_{81}$   &  $\mathcal{L}_{93}$ \\\hline
 1   &    $\mathcal{L}_{145}$  &    $\mathcal{L}_{143}$    & $ \mathcal{L}_{141}$  &  $\mathcal{L}_{154}$         &1   &    $\mathcal{L}_{207}$  &    $\mathcal{L}_{211}$    & $ \mathcal{L}_{204}$  &  $\mathcal{L}_{215}$  & 1   &    $\mathcal{L}_{23}$  &    $\mathcal{L}_{30}$    & $ \mathcal{L}_{17}$   &  $\mathcal{L}_{22}$     & 1   &    $\mathcal{L}_{86}$              &    $\mathcal{L}_{90}$                & $ \mathcal{L}_{82}$   &  $\mathcal{L}_{94}$ \\\hline
 2   &    $\mathcal{L}_{148}$  &    $\mathcal{L}_{150}$    & $\mathcal{L}_{151}$   &  $\mathcal{L}_{155}$         &2   &    $\mathcal{L}_{205}$  &    $\mathcal{L}_{206}$    & $\mathcal{L}_{7}$   &  $\mathcal{L}_{216}$  & 2   &    $\mathcal{L}_{26}$  &    $\mathcal{L}_{27}$    & $ \mathcal{L}_{19}$   &  $\mathcal{L}_{31}$     & 2   &    $\mathcal{L}_{83}$              &    $\mathcal{L}_{87}$                 & $ \mathcal{L}_{5}$   &  $\mathcal{L}_{79}$ \\\hline
 3   &    $\mathcal{L}_{147}$  &    $\mathcal{L}_{149}$    & $\mathcal{L}_{152}$   &  $\mathcal{L}_{156}$         &3   &    $\mathcal{L}_{3}$  &    $\mathcal{L}_{209}$    & $\mathcal{L}_{213}$   &  $\mathcal{L}_{217}$  & 3   &    $\mathcal{L}_{25}$  &    $\mathcal{L}_{28}$    & $ \mathcal{L}_{20}$   &  $\mathcal{L}_{32}$     & 3   &    $\mathcal{L}_{84}$              &    $\mathcal{L}_{88}$                & $ \mathcal{L}_{1}$   &  $\mathcal{L}_{80}$ \\\noalign{\hrule height 1.pt}

 $x_C=3$   &\multirow{2}{*}{0} & \multirow{2}{*}{1} & \multirow{2}{*}{2} & \multirow{2}{*}{3}          & $x_C=3$ & \multirow{2}{*}{0} & \multirow{2}{*}{1} & \multirow{2}{*}{2} & \multirow{2}{*}{3}                                                & $x_C=3$ & \multirow{2}{*}{0} & \multirow{2}{*}{1} & \multirow{2}{*}{2} & \multirow{2}{*}{3}                                               &  $x_C=3$ & \multirow{2}{*}{0} & \multirow{2}{*}{1} & \multirow{2}{*}{2} & \multirow{2}{*}{3} \\
 $x_D=0$   &   &   &   &                                                                               & $x_D=1$ &   &   &   &                                                                                                                      & $x_D=2$ &   &   &   &                                                                                                                     &  $x_D=3$ &   &   &   &   \\\noalign{\hrule height 0.9pt}
 0   &    $\mathcal{L}_{161}$  &    $\mathcal{L}_{138}$    & $ \mathcal{L}_{157}$  &  $\mathcal{L}_{199}$ & 0   &   $\mathcal{L}_{212}$  &    $\mathcal{L}_{222}$   & $\mathcal{L}_{220}$    &  $\mathcal{L}_{218}$     & 0   &    $\mathcal{L}_{39}$  &    $\mathcal{L}_{43}$    & $\mathcal{L}_{33}$      &  $\mathcal{L}_{77}$          &0   &    $\mathcal{L}_{96}$  &    $\mathcal{L}_{100}$    & $ \mathcal{L}_{95}$   &  $\mathcal{L}_{275}$ \\\hline 
 1   &    $\mathcal{L}_{162}$  &    $\mathcal{L}_{165}$    & $ \mathcal{L}_{158}$  &  $\mathcal{L}_{200}$ & 1   &   $\mathcal{L}_{219}$  &    $\mathcal{L}_{262}$   & $ \mathcal{L}_{227}$   &  $\mathcal{L}_{266}$     & 1   &    $\mathcal{L}_{40}$  &    $\mathcal{L}_{44}$    & $ \mathcal{L}_{34}$     &  $\mathcal{L}_{78}$          &1   &    $\mathcal{L}_{99}$  &    $\mathcal{L}_{101}$    & $ \mathcal{L}_{274}$   &  $\mathcal{L}_{276}$ \\\hline 
 2   &    $\mathcal{L}_{159}$  &    $\mathcal{L}_{163}$    & $\mathcal{L}_{197}$   &  $\mathcal{L}_{201}$ & 2   &    $\mathcal{L}_{210}$  &    $\mathcal{L}_{263}$    & $\mathcal{L}_{228}$    &  $\mathcal{L}_{267}$    & 2   &    $\mathcal{L}_{37}$  &    $\mathcal{L}_{41}$    & $ \mathcal{L}_{36}$     &  $\mathcal{L}_{75}$          &2   &    $\mathcal{L}_{97}$  &    $\mathcal{L}_{102}$    & $ \mathcal{L}_{91}$   &  $\mathcal{L}_{277}$ \\\hline 
 3   &    $\mathcal{L}_{160}$  &    $\mathcal{L}_{164}$    & $\mathcal{L}_{198}$   &  $\mathcal{L}_{202}$ & 3   &    $\mathcal{L}_{221}$  &    $\mathcal{L}_{264}$    & $\mathcal{L}_{265}$    &  $\mathcal{L}_{268}$    & 3   &    $\mathcal{L}_{38}$  &    $\mathcal{L}_{42}$    & $\mathcal{L}_{35}$      &  $\mathcal{L}_{76}$          &3   &    $\mathcal{L}_{98}$  &    $\mathcal{L}_{273}$    & $ \mathcal{L}_{92}$   &  $\mathcal{L}_{278}$ \\\noalign{\hrule height 1.pt}
\end{tabular}
\caption{Example 1: Latin-Hyper Cube representing the Relay Map where $x_E=2$, entries $x_A$'s and $x_B$'s are as mentioned, $x_C$'s entries are along the rows and $x_D$'s entries are along the columns of each $4\times 4$ matrix for fixed values of $x_A$ and $x_B$.}
\vspace{-.3cm}
\end{table*}

\begin{table*}[ht]
\vspace{-.2 cm}
\tiny
\centering
\renewcommand{\arraystretch}{1.2}
\begin{tabular}{!{\vrule width 1.pt}c!{\vrule width 0.9pt}p{0.32cm}|p{0.32cm}|p{0.32cm}|p{0.3cm}!{\vrule width 1.pt}c!{\vrule width 0.9pt}p{0.32cm}|p{0.32cm}|p{0.32cm}|p{0.32cm}!{\vrule width 1.pt}c!{\vrule width 0.9pt}p{0.32cm}|p{0.32cm}|p{0.32cm}|p{0.32cm}!{\vrule width 1.pt}c!{\vrule width 0.9pt}p{0.32cm}|p{0.32cm}|p{0.32cm}|p{0.32cm}!{\vrule width 1.pt}}\noalign{\hrule height 1.pt}
 $x_C=0$   &\multirow{2}{*}{0} & \multirow{2}{*}{1} & \multirow{2}{*}{2} & \multirow{2}{*}{3}                  & $x_C=0$ & \multirow{2}{*}{0} & \multirow{2}{*}{1} & \multirow{2}{*}{2} & \multirow{2}{*}{3}                                    & $x_C=0$ & \multirow{2}{*}{0} & \multirow{2}{*}{1} & \multirow{2}{*}{2} & \multirow{2}{*}{3}            &  $x_C=0$ & \multirow{2}{*}{0} & \multirow{2}{*}{1} & \multirow{2}{*}{2} & \multirow{2}{*}{3} \\
 $x_D=0$   &   &   &   &                                                                                       & $x_D=1$ &   &   &   &                                                                                                          & $x_D=2$ &   &   &   &                                                                                  &  $x_D=3$ &   &   &   &   \\\noalign{\hrule height 0.9pt}
   0       &    $\mathcal{L}_{228}$  &    $\mathcal{L}_{196}$    & $ \mathcal{L}_{235}$   &  $\mathcal{L}_{232}$   & 0   &    $\mathcal{L}_{252}$  &    $\mathcal{L}_{276}$    & $\mathcal{L}_{281}$   &  $\mathcal{L}_{184}$               & 0   &    $\mathcal{L}_{132}$    &    $\mathcal{L}_{120}$    & $\mathcal{L}_{266}$     &  $\mathcal{L}_{124}$  &  0       &    $\mathcal{L}_{312}$             &    $\mathcal{L}_{60}$    & $ \mathcal{L}_{68}$   &  $\mathcal{L}_{64}$ \\\hline
   1       &    $\mathcal{L}_{241}$  &    $\mathcal{L}_{195}$    & $ \mathcal{L}_{236}$   &  $\mathcal{L}_{231}$   & 1   &    $\mathcal{L}_{272}$  &    $\mathcal{L}_{275}$    & $ \mathcal{L}_{282}$    &  $\mathcal{L}_{283}$              & 1   &   $\mathcal{L}_{131}$    &    $\mathcal{L}_{119}$    & $ \mathcal{L}_{267}$    &  $\mathcal{L}_{123}$  &  1       &    $\mathcal{L}_{313}$             &    $\mathcal{L}_{59}$    & $ \mathcal{L}_{67}$   &  $\mathcal{L}_{63}$ \\\hline
   2       &    $\mathcal{L}_{239}$  &    $\mathcal{L}_{240}$    & $ \mathcal{L}_{250}$   &  $\mathcal{L}_{246}$   & 2   &    $\mathcal{L}_{189}$  &    $\mathcal{L}_{278}$   & $ \mathcal{L}_{229}$   &  $\mathcal{L}_{185}$              & 2   &    $\mathcal{L}_{130}$    &    $\mathcal{L}_{134}$    & $ \mathcal{L}_{268}$    &  $\mathcal{L}_{301}$   &  2    &    $\mathcal{L}_{72}$             &    $\mathcal{L}_{70}$    & $ \mathcal{L}_{106}$   &  $\mathcal{L}_{66}$ \\\hline
   3       &    $\mathcal{L}_{227}$  &    $\mathcal{L}_{251}$    & $ \mathcal{L}_{249}$   &  $\mathcal{L}_{234}$   & 3   &    $\mathcal{L}_{179}$  &    $\mathcal{L}_{277}$   & $\mathcal{L}_{186}$    &  $\mathcal{L}_{182}$              & 3   &    $\mathcal{L}_{129}$    &    $\mathcal{L}_{133}$    & $\mathcal{L}_{300}$     &  $\mathcal{L}_{302}$   &  3    &    $\mathcal{L}_{73}$             &    $\mathcal{L}_{69}$    & $ \mathcal{L}_{105}$   &  $\mathcal{L}_{65}$ \\\noalign{\hrule height 1.pt}
 $x_C=1$   &\multirow{2}{*}{0} & \multirow{2}{*}{1} & \multirow{2}{*}{2} & \multirow{2}{*}{3}                                   & $x_C=1$ & \multirow{2}{*}{0} & \multirow{2}{*}{1} & \multirow{2}{*}{2} & \multirow{2}{*}{3}                & $x_C=1$ & \multirow{2}{*}{0} & \multirow{2}{*}{1} & \multirow{2}{*}{2} & \multirow{2}{*}{3}                   &  $x_C=1$ & \multirow{2}{*}{0} & \multirow{2}{*}{1} & \multirow{2}{*}{2} & \multirow{2}{*}{3} \\
 $x_D=0$   &   &   &   &                                                                                                        & $x_D=1$ &   &   &   &                                                                                      & $x_D=2$ &   &   &   &                                                                                         &  $x_D=3$ &   &   &   &   \\\noalign{\hrule height 0.9pt}
   0   &    $\mathcal{L}_{192}$    &    $\mathcal{L}_{242}$    & $ \mathcal{L}_{256}$   &  $\mathcal{L}_{188}$    & 0   &    $\mathcal{L}_{178}$  &    $\pmb{\mathcal{L}}_{\pmb{5}}$    & $ \mathcal{L}_{287}$   &  $\mathcal{L}_{289}$    &0   &    $\mathcal{L}_{303}$       &    $\pmb{\mathcal{L}}_{\pmb{7}}$    & $ \mathcal{L}_{112}$   &  $\mathcal{L}_{108}$           & 0   &    $\mathcal{L}_{314}$            &    $\mathcal{L}_{56}$              & $ \mathcal{L}_{52}$   &  $\mathcal{L}_{318}$ \\\hline
   1   &    $\mathcal{L}_{191}$    &    $\mathcal{L}_{253}$    & $ \mathcal{L}_{255}$   &  $\mathcal{L}_{187}$    & 1   &    $\mathcal{L}_{176}$  &    $\pmb{\mathcal{L}}_{\pmb{1}}$    & $ \mathcal{L}_{183}$   &  $\mathcal{L}_{290}$    &1   &    $\mathcal{L}_{304}$       &    $\pmb{\mathcal{L}}_{\pmb{3}}$    & $ \mathcal{L}_{111}$   &  $\mathcal{L}_{107}$           & 1   &    $\mathcal{L}_{315}$            &    $\mathcal{L}_{55}$              & $ \mathcal{L}_{51}$   &  $\mathcal{L}_{138}$ \\\hline
   2   &    $\mathcal{L}_{194}$    &    $\mathcal{L}_{180}$    & $ \mathcal{L}_{233}$   &  $\mathcal{L}_{248}$    & 2   &    $\mathcal{L}_{284}$  &    $\mathcal{L}_{286}$    & $ \mathcal{L}_{181}$   &  $\mathcal{L}_{174}$    &2   &    $\mathcal{L}_{118}$       &    $\mathcal{L}_{114}$    & $ \mathcal{L}_{122}$   &  $\mathcal{L}_{110}$           & 2   &    $\mathcal{L}_{58}$            &    $\mathcal{L}_{316}$              & $ \mathcal{L}_{62}$   &  $\mathcal{L}_{48}$ \\\hline
   3   &    $\mathcal{L}_{193}$    &    $\mathcal{L}_{190}$    & $ \mathcal{L}_{230}$   &  $\mathcal{L}_{247}$    & 3   &    $\mathcal{L}_{285}$  &    $\mathcal{L}_{175}$    & $ \mathcal{L}_{288}$  &  $\mathcal{L}_{173}$    &3   &    $\mathcal{L}_{117}$       &    $\mathcal{L}_{113}$    & $ \mathcal{L}_{121}$   &  $\mathcal{L}_{109}$           & 3   &    $\mathcal{L}_{57}$              &    $\mathcal{L}_{317}$            & $ \mathcal{L}_{61}$   &  $\mathcal{L}_{47}$ \\\noalign{\hrule height 1.pt}
 $x_C=2$   &\multirow{2}{*}{0} & \multirow{2}{*}{1} & \multirow{2}{*}{2} & \multirow{2}{*}{3}                                            & $x_C=2$ & \multirow{2}{*}{0} & \multirow{2}{*}{1} & \multirow{2}{*}{2} & \multirow{2}{*}{3}            & $x_C=2$ & \multirow{2}{*}{0} & \multirow{2}{*}{1} & \multirow{2}{*}{2} & \multirow{2}{*}{3}                                    &  $x_C=2$ & \multirow{2}{*}{0} & \multirow{2}{*}{1} & \multirow{2}{*}{2} & \multirow{2}{*}{3} \\
 $x_D=0$   &   &   &   &                                                                                                                 & $x_D=1$ &   &   &   &                                                                                  & $x_D=2$ &   &   &   &                                                                                                          &  $x_D=3$ &   &   &   &   \\\noalign{\hrule height 0.9pt}
 0   &    $\mathcal{L}_{254}$  &    $\mathcal{L}_{259}$    & $ \mathcal{L}_{262}$  &  $\mathcal{L}_{263}$         &0   &    $\mathcal{L}_{165}$  &    $\mathcal{L}_{162}$    & $ \mathcal{L}_{200}$  &  $\mathcal{L}_{158}$  & 0   &    $\mathcal{L}_{102}$  &    $\mathcal{L}_{99}$    & $ \mathcal{L}_{307}$   &  $\mathcal{L}_{309}$     & 0   &    $\mathcal{L}_{44}$              &    $\mathcal{L}_{40}$               & $ \mathcal{L}_{78}$   &  $\mathcal{L}_{34}$ \\\hline
 1   &    $\mathcal{L}_{257}$  &    $\mathcal{L}_{212}$    & $ \mathcal{L}_{261}$  &  $\mathcal{L}_{265}$         &1   &    $\mathcal{L}_{291}$  &    $\mathcal{L}_{161}$    & $ \mathcal{L}_{199}$  &  $\mathcal{L}_{157}$  & 1   &    $\mathcal{L}_{100}$  &    $\mathcal{L}_{306}$    & $ \mathcal{L}_{12}$   &  $\mathcal{L}_{310}$     & 1   &    $\mathcal{L}_{43}$              &    $\mathcal{L}_{39}$                & $ \mathcal{L}_{10}$   &  $\mathcal{L}_{33}$ \\\hline
 2   &    $\mathcal{L}_{218}$  &    $\mathcal{L}_{260}$    & $\mathcal{L}_{264}$   &  $\mathcal{L}_{273}$         &2   &    $\mathcal{L}_{164}$  &    $\mathcal{L}_{160}$    & $\mathcal{L}_{8}$   &  $\mathcal{L}_{198}$  & 2   &    $\mathcal{L}_{305}$  &    $\mathcal{L}_{98}$    & $ \mathcal{L}_{16}$   &  $\mathcal{L}_{92}$     & 2   &    $\mathcal{L}_{42}$              &    $\mathcal{L}_{38}$                 & $ \mathcal{L}_{14}$   &  $\mathcal{L}_{319}$ \\\hline
 3   &    $\mathcal{L}_{258}$  &    $\mathcal{L}_{210}$    & $\mathcal{L}_{219}$   &  $\mathcal{L}_{274}$         &3   &    $\mathcal{L}_{163}$  &    $\mathcal{L}_{159}$    & $\mathcal{L}_{4}$   &  $\mathcal{L}_{197}$  & 3   &    $\mathcal{L}_{101}$  &    $\mathcal{L}_{97}$    & $ \mathcal{L}_{308}$   &  $\mathcal{L}_{92}$     & 3   &    $\mathcal{L}_{41}$              &    $\mathcal{L}_{37}$                & $ \mathcal{L}_{2}$   &  $\mathcal{L}_{320}$ \\\noalign{\hrule height 1.pt}

 $x_C=3$   &\multirow{2}{*}{0} & \multirow{2}{*}{1} & \multirow{2}{*}{2} & \multirow{2}{*}{3}          & $x_C=3$ & \multirow{2}{*}{0} & \multirow{2}{*}{1} & \multirow{2}{*}{2} & \multirow{2}{*}{3}                                                & $x_C=3$ & \multirow{2}{*}{0} & \multirow{2}{*}{1} & \multirow{2}{*}{2} & \multirow{2}{*}{3}                                               &  $x_C=3$ & \multirow{2}{*}{0} & \multirow{2}{*}{1} & \multirow{2}{*}{2} & \multirow{2}{*}{3} \\
 $x_D=0$   &   &   &   &                                                                               & $x_D=1$ &   &   &   &                                                                                                                      & $x_D=2$ &   &   &   &                                                                                                                     &  $x_D=3$ &   &   &   &   \\\noalign{\hrule height 0.9pt}
 0   &    $\mathcal{L}_{211}$  &    $\mathcal{L}_{207}$    & $ \mathcal{L}_{270}$  &  $\mathcal{L}_{204}$ & 0   &   $\mathcal{L}_{150}$  &    $\mathcal{L}_{294}$   & $\mathcal{L}_{148}$    &  $\mathcal{L}_{298}$     & 0   &    $\mathcal{L}_{90}$  &    $\mathcal{L}_{86}$    & $\mathcal{L}_{94}$      &  $\mathcal{L}_{83}$          &0   &    $\mathcal{L}_{30}$  &    $\mathcal{L}_{26}$    & $ \mathcal{L}_{22}$   &  $\mathcal{L}_{115}$ \\\hline 
 1   &    $\mathcal{L}_{208}$  &    $\mathcal{L}_{166}$    & $ \mathcal{L}_{269}$  &  $\mathcal{L}_{203}$ & 1   &   $\mathcal{L}_{149}$  &    $\mathcal{L}_{295}$   & $ \mathcal{L}_{147}$   &  $\mathcal{L}_{299}$     & 1   &    $\mathcal{L}_{89}$  &    $\mathcal{L}_{13}$    & $ \mathcal{L}_{93}$     &  $\mathcal{L}_{81}$          &1   &    $\mathcal{L}_{29}$  &    $\mathcal{L}_{25}$    & $ \mathcal{L}_{21}$   &  $\mathcal{L}_{15}$ \\\hline 
 2   &    $\mathcal{L}_{209}$  &    $\mathcal{L}_{224}$    & $\mathcal{L}_{214}$   &  $\mathcal{L}_{279}$ & 2   &    $\mathcal{L}_{292}$  &    $\mathcal{L}_{296}$    & $\mathcal{L}_{156}$    &  $\mathcal{L}_{142}$    & 2   &    $\mathcal{L}_{88}$  &    $\mathcal{L}_{84}$    & $ \mathcal{L}_{80}$     &  $\mathcal{L}_{311}$          &2   &    $\mathcal{L}_{28}$  &    $\mathcal{L}_{321}$    & $ \mathcal{L}_{32}$   &  $\mathcal{L}_{18}$ \\\hline 
 3   &    $\mathcal{L}_{206}$  &    $\mathcal{L}_{205}$    & $\mathcal{L}_{271}$   &  $\mathcal{L}_{280}$ & 3   &    $\mathcal{L}_{293}$  &    $\mathcal{L}_{297}$    & $\mathcal{L}_{155}$    &  $\mathcal{L}_{141}$    & 3   &    $\mathcal{L}_{87}$  &    $\mathcal{L}_{83}$    & $\mathcal{L}_{79}$      &  $\mathcal{L}_{9}$          &3   &    $\mathcal{L}_{27}$  &    $\mathcal{L}_{322}$    & $ \mathcal{L}_{31}$   &  $\mathcal{L}_{11}$ \\\noalign{\hrule height 1.pt}
\end{tabular}
\caption{Example 1: Latin-Hyper Cube representing the Relay Map where $x_E=3$, entries $x_A$'s and $x_B$'s are as mentioned, $x_C$'s entries are along the rows and $x_D$'s entries are along the columns of each $4\times 4$ matrix for fixed values of $x_A$ and $x_B$.}
\vspace{-1cm}
\end{table*}


\begin{lemma}
The clustering map used at the relay node R cannot remove the singular fade spaces $\left\langle \left[ \Delta x_1, \Delta x_2,..., \Delta x_n \right]\right\rangle $ where at least one of $\Delta x_k=0$ for some $k=1,2,...,n$ and simultaneously satisfy the mutually exclusive law.
\end{lemma}
\begin{proof} Let $\mathcal{S}=\left\langle \left[x_1-x'_1, x_2-x'_2, ..., x_n-x'_n\right]\right\rangle^{ \bot} $ be a singular fade state where for some $1 \leq k \leq n$, $x_k-x'_k=0$. Then, in order to remove $\mathcal{S}$, $ ( x_1,  x_2,..., x_k,..., x_n) $ and $ ( x'_1,  x'_2,..., x_k,..., x'_n) $ that must be kept in the same cluster. This would imply user $X_k$ not being able to distinguish between the messages $x_l$ and $x'_l$ for some $1 \leq l \leq n, ~l \neq k,~x_l \neq x'_l$ sent by user $X_l$. This will clearly violate the mutually exclusive law, since in order to satisfy the mutually exclusive law, for the same value of $x_k$, all possible $n$-tuples of messages must be kept in different clusters. These two statements cannot be satisfied at the same time, hence such a singular fade subspace cannot be removed if the mutually exclusive law has to be satisfied by the relay map used in the BC phase. 
\end{proof}
We refer to the singular fade subspaces whose harmful effects cannot be removed by a proper choice of the clustering, as the \textit{non-removable singular fade subspaces} also talked about in \cite{MuR}. 

\begin{corollary} There are $\left[ (\frac{M}{2})^n - (\frac{M}{2}) +1\right] M^{n-1}$ Removable and $\sum^{n-1}_{k=1} (^{n}_{k}) \left[ (\frac{M}{2})^k - (\frac{M}{2}) +1\right] M^{k-1}$ Non-Removable Singular Fade Subspaces for $n$-way relaying when $M$-PSK constellation is used at the end nodes.
\end{corollary}

\begin{corollary} The number of non-removable singular subspaces is $O(M^{n-1})$ while the number of removable singular fade subspaces is $O(M^{n})$. 
\end{corollary}

Thus, the number of non-removable singular fade subspaces is a small fraction of the total number of singular fade subspaces. For the five-way relaying scenario described in the previous section, there are $13981$ singular fade subspaces for five-way relaying, out of which the scheme given in this paper removes $7936$ singular fade subspaces using $5$-fold Latin Hyper-Cubes of side $4$. This can be done, as described above, by first marking the singularity removal constraints in the empty $4 \times 4 \times 4 \times 4 \times 4 $ array and then completing the array to form a $5$-fold Latin Hyper-Cubes of side $4$ using Algorithm 1. We now illustrate the removal of a singular fade state for five-way relaying with the help of the following example.


\begin{example} 
Consider the case for five-way relaying scenario where 4-PSK is used at end nodes A, B, C, D and E, and one of the singular fade subspace to be removed is given by,

\begin{tiny}
\noindent\vspace{-.10 cm}
\begin{displaymath}
\mathcal{S}''= \left\langle \left[ {\begin{array}{c}
-1-j\\
-2j\\
-2j\\
1-j\\
1+j\\
\end{array} } \right]\right\rangle ^{\bot}\hspace{-0.3cm}
~=~\hspace{-0.1cm}\left\langle \left[ {\begin{array}{c}
1-j \\
2 \\
2 \\
1+j \\
-1+j \\
\end{array} } \right]\right\rangle ^{\bot}\hspace{-0.3cm}
\end{displaymath}
\begin{displaymath}
~~~~~~~~~~~~~~=~\hspace{-0.1cm}\left\langle \left[ {\begin{array}{c}
1+j \\
2j \\
2j \\
-1+j \\
-1-j \\
\end{array} } \right]\right\rangle ^{\bot}\hspace{-0.3cm}
~=~\hspace{-0.1cm}\left\langle \left[ {\begin{array}{c}
-1+j \\
-2 \\
-2 \\
-1-j \\
1-j \\
\end{array} } \right]\right\rangle ^{\bot}.
\end{displaymath}

\end{tiny}

The first vector is $\left[ -1-j, ~ -2j, ~ -2j, ~ 1-j, ~ 1+j \right]$. Now, $-1-j$ can be obtained either as a difference of $x_A=-1$ and $x'_A=j$ or as a difference of $x_A=-j$ and $x'_A=1$; $-2j$ can be obtained only as a difference of $-j$ and $j$; $1-j$ can be obtained as a difference of $x_D=1$ and $x'_D=j$ or as a difference of $x_D=-j$ and $x'_D=1$; $1+j$ can be obtained as a difference of $x_E=1$ and $x_E=-j$ or as a difference of $x_E=j$ and $x_E=-1$. Thus, the entries corresponding to {\footnotesize $\left\{(-1,-j,-j,1,1),(j,j,j,j,-j)\right\}$, $\left\{(-1,-j,-j,-j,1),(j,j,j,1,-j)\right\}$, $\left\{(-1,-j,-j,1,j),(j,j,j,j,-1)\right\}$, $\left\{(-1,-j,-j,-j,j),(j,j,j,1,-1)\right\}$, $\left\{(-j,-j,-j,1,1),(1,j,j,j,-j)\right\}$, $\left\{(-j,-j,-j,-j,1),(1,j,j,1,-j)\right\}$, $\left\{(-j,-j,-j,1,j),(1,j,j,j,-1)\right\}$, $\left\{(-j,-j,-j,-j,j),(1,j,j,1,-1)\right\}$ } must lie in the same clustering representing the network coding map used at the relay node in the BC phase. Similarly the constraints resulting from the other three vectors above can be obtained. Replacing $1,j,-1,-j$ with $0,1,2,3$, we get constraints in the form of entries that must be kept the same in the $4 \times 4 \times 4 \times 4 \times 4 $ array representing the clustering. For instance, corresponding to {\footnotesize $\left\{(-1,-j,-j,1,1),(j,j,j,j,-j)\right\}$}, the entries $\left\{(2,3,3,0,0),(1,1,1,1,3)\right\}$ are kept in the same cluster $\mathcal{L}_{1}$ as shown in bold in Tables I, II, III, IV. The constrained Hyper-Cube of side 4 is completed using Algorithm 1 to form a 4-fold Latin Hyper-Cube of side 4. The completed Latin Hyper Cube is as shown in Tables I, II, III and IV, where the constraints are marked in bold. 
\end{example}

Similarly, a (removable) singular fade subspace can be removed by first constraining the array representing the relay map, and then completing the constrained array using the provided algorithm.


\section{Simulation Results}
\vspace{-.1cm}
Simulation results presented in this section identify the cases where the proposed scheme outperforms the naive approach that uses the same map for all fade states and vice verse. Here, the channel states are distributed according to Rician distribution and channel variances equal to 0 dB and the frame length is 256 bits. Fig. \ref{fig:comparison_plot} compares the SNR vs bit-error-rate curves for three-way, four-way and five-way relaying scenario for (a) the adaptive network coding scheme presented in this paper with (b) the non-adaptive network coding using two channel uses, in which the same array is used by the relay as an encoder for all channel conditions. The details of three-way and four-way relaying can be found in \cite{SVR} and \cite{ShR}. The non-adaptive network coding for three-way, four-way and five-way relaying utilizes the same $3$-fold Latin Hyper-Cubes of side $4$, $4$-fold Latin Hyper-Cube of side $4$ and $5$-fold Latin Hyper-Cube of side $4$ (respectively) for all channel conditions.

\begin{figure}[ht]
\vspace{-.1 cm}
\hspace{-.9 cm}
\includegraphics[height=64mm,width=10.2cm]{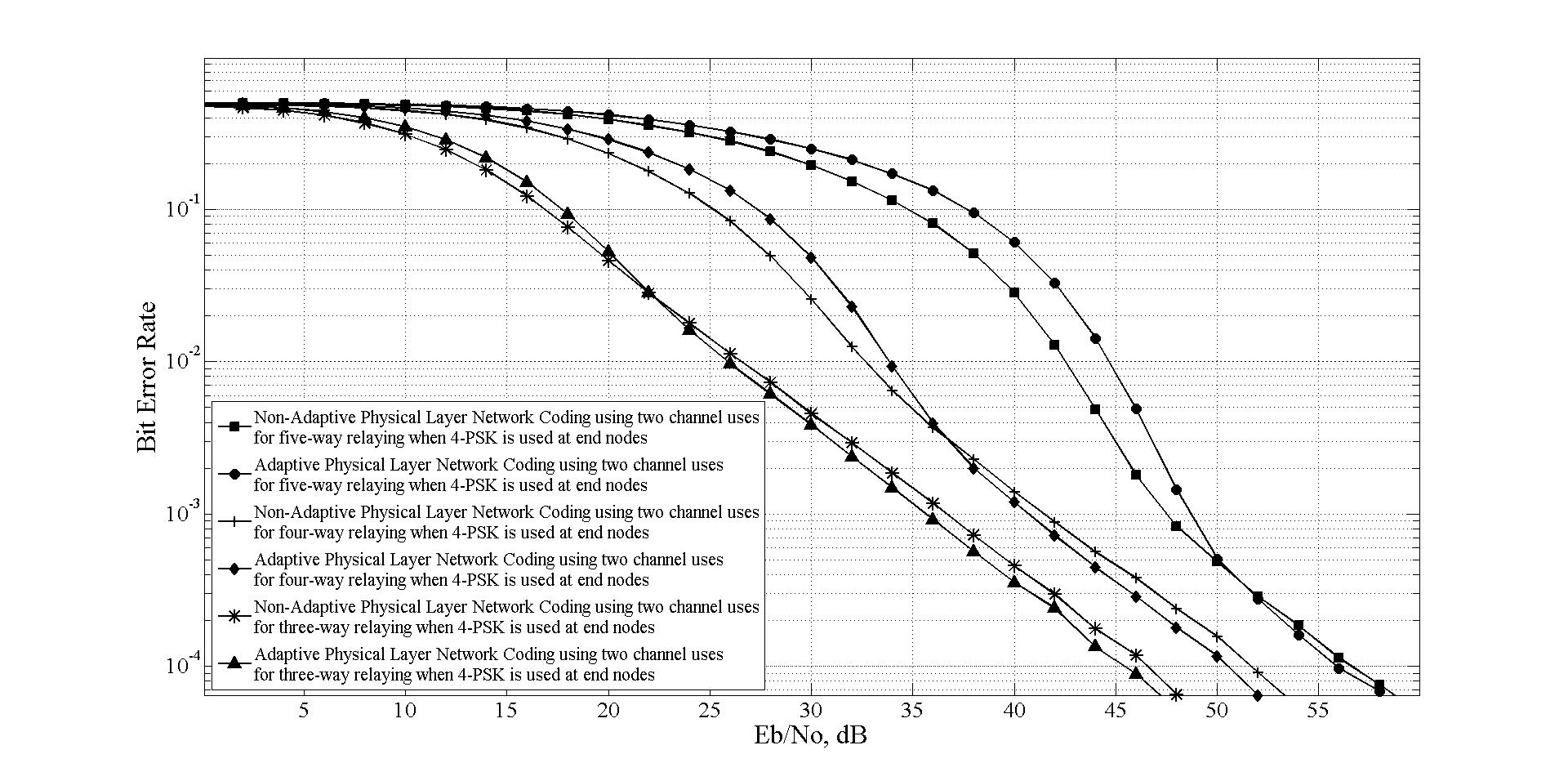}
\vspace{-1 cm}
\caption{SNR vs BER curves for different schemes for multi-way relaying when the Rician Factors is 20 dB}	
\label{fig:comparison_plot}	
\vspace{-.5cm}
\end{figure}

For any communication system, the effects of additive noise are predominant at low SNR and the effects of noise due to fading dominate at high SNR. Since adaptive network coding scheme attempts at reducing the effects of fading for the $n$-way relaying scenario by removing a part of singular fade states, (a) performs better than (b) at high SNR. This also implies that the performance of the scheme in the MA phase (removal of singular fade spaces takes place in the MA phase) is predominant at higher SNRs, and the performance of the scheme in the BC phase dominates at lower SNRs. Also, as the number of user nodes $n$ increases, the SNR at which the performance of adaptive network coding improves over the performance of non-adaptive network coding increases, as can be seen in the plot, since the size of the received constellation in the MA phase and hence the size of the constellation used in the BC phase increases with increasing values of $n$.

%
%
\section{Conclusion}
\vspace{-.cm}
We consider the $n$-way wireless relaying scenario, where $n$ nodes operate in half-duplex mode and transmit points from the same $M$-PSK constellation. Information exchange is made possible using just two channels uses, unlike the existing work done for the case, to the best of our knowledge. The Relay node clusters the $M^n$ possible transmitted tuples $\left(x_{1},x_{2},...,x_n\right)$ into various clusters depending on the fade states such that \textit{the exclusive law} is satisfied and some of the singular fade subspaces are removed. This necessary requirement of satisfying the exclusive law is shown to be the same as the clustering being represented by a $n$-fold Latin Hyper-Cube of side $M$. The size of the clustering utilizing modified clustering may not be the best that can be achieved, and it might be possible to fill the array with lesser symbols. 

\begin{center}
\textsc{Acknowledgments}
\end{center}
\vspace{-.15cm}
The first author would like to thank her colleague Vijayvaradharaj T Muralidharan at the Department of Electrical and Communication Engineering, Indian Institute of Science for his feedback and suggestions related to this paper.

\end{document}